\documentclass[a4paper]{article}
\usepackage{enumitem}
\usepackage{calc}

\usepackage{graphicx} 

\usepackage{amsmath}
\usepackage{amssymb,tikz}
\usepackage{amsthm}
\usepackage{amscd}
\usepackage{mathrsfs}
\usepackage{todonotes}
\usepackage{float}

\usetikzlibrary{calc} 

\usepackage{marvosym}

\newtheorem{theorem}			     {Theorem} [section]
\newtheorem{proposition}[theorem]	 {Proposition}	
\newtheorem{corollary}	  [theorem]	 {Corollary}	
\newtheorem{conjecture}	  [theorem]	 {Conjecture}
\newtheorem{ansatz}	  [theorem]	 {Ansatz}	
\newtheorem{lemma}	      [theorem]  {Lemma}		
\theoremstyle{definition}

\newtheorem{remark} [theorem]  {Remark}

\newcommand{\C}{\mathbb{C}}
\newcommand{\R}{\mathbb{R}}

\newcommand{\e}{\epsilon}

\DeclareMathOperator{\tr}{Tr}

\newcommand{\Ai}{{\rm Ai}}

\def\XXint#1#2#3{{\setbox0=\hbox{$#1{#2#3}{\int}$}
\vcenter{\hbox{$#2#3$}}\kern-.5\wd0}}

\usepackage{tikz}
\usetikzlibrary{arrows}
\usetikzlibrary{decorations.pathmorphing}
\usetikzlibrary{decorations.markings}
\usetikzlibrary{patterns}
\usetikzlibrary{automata}
\usetikzlibrary{positioning}
\usepackage{tikz-cd}

\tikzset{->-/.style={decoration={
				markings,
				mark=at position #1 with {\arrow{latex}}},postaction={decorate}}}
	
	\tikzset{-<-/.style={decoration={
				markings,
				mark=at position #1 with {\arrowreversed{latex}}},postaction={decorate}}}

\usetikzlibrary{shapes.misc}\tikzset{cross/.style={cross out, draw, 
         minimum size=2*(#1-\pgflinewidth), 
         inner sep=0pt, outer sep=0pt}}

\usepackage{pgfplots}

\numberwithin{equation}{section}

\def\be{\begin{equation}}
\def\ee{\end{equation}}

\usepackage[colorlinks=true]{hyperref}
\hypersetup{urlcolor=blue, citecolor=red, linkcolor=blue}

\begin{document}
\title{Large deviations for the log-Gamma polymer}
\author{Tom Claeys and Julian Mauersberger}

\maketitle
\begin{abstract}
We conjecture an explicit expression for the lower tail large deviation rate function of the partition function of the log-Gamma polymer. We rigorously prove our result, except for one step for which we only provide heuristic evidence. We show that the large deviation rate function matches with that of last passage percolation with exponential weights in the zero-temperature limit, and with the lower tail of the Tracy-Widom distribution for moderate deviations.
\end{abstract}

\section{Introduction and main results}
The log-Gamma polymer is a model of directed lattice paths in a random environment, introduced by Seppäläinen \cite{S12}. The random environment is special in the sense that it allows exact expressions for several statistics of interest. In particular, the Laplace transform of the partition function admits an exact characterization in terms of a Fredholm determinant for finite lattice size. This characterization enables one to study the asymptotic behavior of the partition function in the limit of a large lattice. In particular, it was shown in \cite{BCD21, BCR13} that the partition function has Tracy-Widom fluctuations. 
Other polymer models with similar features are the O'Connell-Yor polymer, the mixed polymer, and the continuum directed polymer or Kardar-Parisi-Zhang equation, as well as their zero-temperature counterparts which are last passage percolation or corner growth models.  
While a Fredholm determinant identity for the Laplace transform is effective to study the typical large $n$ behavior, it is much more challenging to extract precise information about the lower tails of the probability distribution and to obtain accurate large deviation estimates. Such estimates have been recently established for the narrow wedge solution of the Kardar-Parisi-Zhang equation \cite{CafassoClaeys, CharlierClaeysRuzza, CorwinGhosal, LeDoussal, Tsai}, but remain to be understood for other exactly solvable polymer models.

In this paper, we conjecture an explicit expression for the large deviation rate function of the log-Gamma polymer partition function. Our approach is based on a connection between the log-Gamma polymer partition function and a signed biorthogonal measure found in \cite{CC24}, which leads to a Fredholm determinant which we analyze asymptotically.  Our analysis is rigorous, except for one step, namely an estimate for the difference between two Fredholm determinants. We provide heuristic evidence for this estimate, and explain the obstacles for a rigorous proof.
We believe that a similar method can lead to the large deviation rate functions in other exactly solvable polymer models, like the inhomogeneous log-Gamma polymer on a rectangular lattice, the O'Connell-Yor polymer \cite{OY}, and the mixed polymer \cite{BCFV}.

\medskip
\paragraph{Definition of the model.}
Consider the square lattice $\{1,2, \ldots, n\}\times \{1,2, \ldots, n\}$. To each point $(i,j)$ in the grid we assign an independent positive random variable $d_{i,j}(\theta)$, which is \emph{inverse-Gamma distributed}, that is, it has probability density
\begin{align*}
 \frac{x^{-2\theta-1}}{\Gamma(2\theta)} e^{-1/x},\qquad x>0,
\end{align*} 
depending on a parameter $\theta>0$.

Given an up-right path $\pi$ connecting $(1,1)$ with $(n,n)$, we define its weight as \[w_{n,\theta}(\pi)=\prod_{(i,j) \in \pi} d_{i,j}(\theta).\] The \emph{log-Gamma polymer partition function} is defined as the sum of the weights of all such paths:
\begin{align}\label{def:partition function}
	Z_n(\theta) =\sum_{\pi: (1,1) \to (n,n)} w_{n,\theta}(\pi)= \sum_{\pi: (1,1) \to (n,n)} \prod_{(i,j) \in \pi} d_{i,j}(\theta),
\end{align}
where the sum is taken over all up-right directed paths $\pi$ in the grid $\{1,\ldots,n\}^2$ starting at $(1,1)$ and ending at $(n,n)$, as illustrated in Figure \ref{figure: LogGamma}.  
The model can be generalized by considering a rectangular instead of a square lattice, and by allowing different parameters $\theta=\theta_{i,j}$ for different vertices $(i,j)$, but we will restrict ourselves to the square lattice with homogeneous weights here.
\begin{figure}[t]
\begin{center}
    \setlength{\unitlength}{1truemm}
    \begin{picture}(100,50)(0,-10)
    \put(5,2){$(1,1)$} \put(55,52){$(n,n)$}
    \put(10,0){\thicklines\circle*{.8}}
\put(20,0){\thicklines\circle*{.8}}
\put(30,0){\thicklines\circle*{.8}}
\put(40,0){\thicklines\circle*{.8}}
\put(50,0){\thicklines\circle*{.8}}
\put(60,0){\thicklines\circle*{.8}}    
     
  \put(10,10){\thicklines\circle*{.8}}
\put(20,10){\thicklines\circle*{.8}}
\put(30,10){\thicklines\circle*{.8}}
\put(40,10){\thicklines\circle*{.8}}
\put(50,10){\thicklines\circle*{.8}}
\put(60,10){\thicklines\circle*{.8}}

    \put(10,20){\thicklines\circle*{.8}}
\put(20,20){\thicklines\circle*{.8}}
\put(30,20){\thicklines\circle*{.8}}
\put(40,20){\thicklines\circle*{.8}}
\put(50,20){\thicklines\circle*{.8}}
\put(60,20){\thicklines\circle*{.8}}    
    
    \put(10,30){\thicklines\circle*{.8}}
\put(20,30){\thicklines\circle*{.8}}
\put(30,30){\thicklines\circle*{.8}}
\put(40,30){\thicklines\circle*{.8}}
\put(50,30){\thicklines\circle*{.8}}
\put(60,30){\thicklines\circle*{.8}}    
    
\put(10,40){\thicklines\circle*{.8}}
\put(20,40){\thicklines\circle*{.8}}
\put(30,40){\thicklines\circle*{.8}}
\put(40,40){\thicklines\circle*{.8}}
\put(50,40){\thicklines\circle*{.8}}
\put(60,40){\thicklines\circle*{.8}}    
    
    \put(10,50){\thicklines\circle*{.8}}
\put(20,50){\thicklines\circle*{.8}}
\put(30,50){\thicklines\circle*{.8}}
\put(40,50){\thicklines\circle*{.8}}
\put(50,50){\thicklines\circle*{.8}}
\put(60,50){\thicklines\circle*{.8}}

    \put(10,0){\line(1,0){10}}\put(16.7,0){\thicklines\vector(1,0){.0001}}
    \put(20,10){\line(1,0){10}}\put(26.7,10){\thicklines\vector(1,0){.0001}}
    \put(30,10){\line(1,0){10}}\put(36.7,10){\thicklines\vector(1,0){.0001}}
    \put(40,20){\line(1,0){10}}\put(46.7,20){\thicklines\vector(1,0){.0001}}
    \put(50,40){\line(1,0){10}}\put(56.7,40){\thicklines\vector(1,0){.0001}}

    \put(20,0){\line(0,1){10}}\put(20,6.7){\thicklines\vector(0,1){.0001}}
    \put(40,10){\line(0,1){10}}\put(40,16.7){\thicklines\vector(0,1){.0001}}
    \put(50,20){\line(0,1){10}}\put(50,26.7){\thicklines\vector(0,1){.0001}}
    \put(50,30){\line(0,1){10}}\put(50,36.7){\thicklines\vector(0,1){.0001}}
    \put(60,40){\line(0,1){10}}\put(60,46.7){\thicklines\vector(0,1){.0001}}      
    \end{picture}
    \caption{The log-Gamma polymer lattice for $n=6$, with a possible up-right path.}
    \label{figure: LogGamma}
\end{center}
\end{figure}
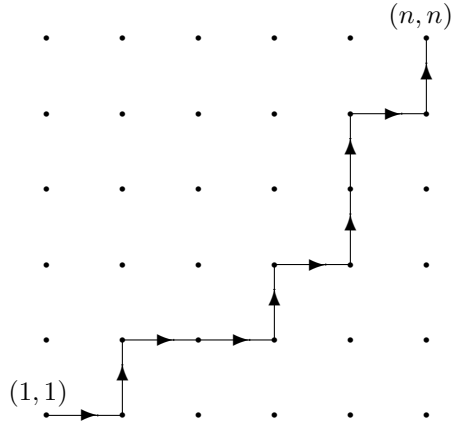

\medskip

For $\theta>0$ small, the density of the weights $d_{i,j}(\theta)$ has a heavy tail, such that one expects large variations between the weights of different paths, implying that the partition function will be dominated by few paths with large weights for large $n$. 
More precisely, as $\theta\to 0$, the random variables $u_{i,j}(\theta):=2\theta\log d_{i,j}(\theta)$ converge in distribution to independent exponential random variables $u_{i,j}$, i.e., $\mathbb P(u_{i,j}\leq s)=1-e^{-u_{i,j}}$. As $\theta\to 0$ (with $n$ fixed), we then obtain the weak convergence \cite{S12}
\[2\theta\log Z_n(\theta)\longrightarrow F_n^{\rm LPP}:=\max_{\pi: (1,1) \to (n,n)}\sum_{(i,j)\in \pi}u_{i,j}.\]
The right hand side is the maximum additive weight of an up-right path in random exponential environment. This model is known as {\em last passage percolation} or corner growth with exponential weights, see \cite{BorodinPeche,DiekerWarren,Johansson}.
From the physics point of view, $\theta>0$ is a parameter proportional to the temperature of the model. 

\paragraph{Known asymptotic results.}
Seppäläinen \cite{S12} proved that the average of $\log Z_n(\theta)$ behaves for large $n$ like $-2n\psi(\theta)$, where $\psi=\frac{\Gamma'}{\Gamma}$ is the di-Gamma function, and that its variance is of order $n^{2/3}$.
Later, in 2013, Borodin, Corwin, and Remenik \cite[Corollary 1.8]{BCR13} showed that the Laplace transform of $Z_{n}(\theta)$ is equal to a Fredholm determinant, see also \cite{BC}: we have for $u\in\mathbb C$ with $\Re u>0$ that
\begin{align}\label{eq:detid1}
	\mathbb{E}(e^{-u Z_n(\theta)}) = \det(I + K_{n}^{u,\theta})_{L^2(\Sigma)},
\end{align}
where $\Sigma$ is a positively oriented circle in the complex plane around $-\theta$ with radius $<\theta$, and  $K_{n}^{u,\theta}:L^2(\Sigma)  \to L^2(\Sigma)$ is the integral kernel operator with kernel
\[K_{n}^{u,\theta}(v,v') := \frac{1}{(2 \pi i)^2} \int_{\ell} d w \frac{\pi u^{w-v}}{\sin \pi(v - w)} \frac{W(w)}{W(v)}\frac{1}{w - v'},\qquad v,v'\in{\Sigma},\]
where $\ell$ is a vertical line at the right of $\Sigma$ and at the left of $\theta$, and
\begin{align} \label{def:W}
	W(z)= W(z;n,\theta) = \frac{ \Gamma(\theta-z)^n}{ \Gamma(\theta+z)^n}.
\end{align}
We recall that the Fredholm determinant can be evaluated through the {\em Fredholm series}
\begin{equation}\label{def:Fredholm series}\det(I + K)_{L^2(\gamma)}=1+\sum_{k=1}^\infty \frac{1}{k!}\int_{\gamma^k}\det\left(K(x_i,x_j)\right)_{i,j=1}^k dx_1\cdots dx_k.\end{equation}
The above Fredholm determinant identity combined with a saddle point analysis of the kernel $K_{n}^{u,\theta}$ allowed the authors to prove that the fluctuations of $\log Z_n(\theta)$ are described by the {\em Tracy-Widom distribution}. 
More precisely \cite[Theorem 1]{BCR13},
\begin{align} \label{eq: central limit theorem}
	\lim_{n \to \infty} \mathbb{P} \left(\frac{\log Z_n(\theta) + 2n \psi(\theta)}{(-\psi''(\theta)n)^{1/3}}\le r \right) = F_{\text{TW}}\left(r\right),\qquad r\in\mathbb R
\end{align}
for sufficiently small $\theta > 0$, where $F_{\rm TW}(r)$ is the $\beta=2$ Tracy-Widom distribution \cite{TW94}. This distribution describes among others the largest eigenvalue distribution of large Hermitian random matrices with soft edges, and can be written as an Airy kernel Fredholm determinant,
\[F_{\rm TW}(r)=\det\left(1-K^{\Ai}\right)_{L^2(r,\infty)},\qquad K^{\Ai}(x,x')=\frac{\Ai(x)\Ai'(x')-\Ai(x')\Ai'(x)}{x-x'}.\]
This result was subsequently generalized in \cite{KQ18, BCD21}. 
One can see \eqref{eq: central limit theorem} as the non-zero temperature analogue of a well-known limit theorem for last passage percolation with exponential weights \cite{Johansson}, stating that 
\begin{align} \label{eq: central limit theorem0temp}
	\lim_{n \to \infty} \mathbb{P} \left(\frac{F_n^{\rm LPP} - 4 n  }{2^{4/3}n^{1/3}}\le r \right) = F_{\text{TW}}\left(r\right),\qquad r\in\mathbb R.
\end{align}

\medskip

While \eqref{eq: central limit theorem} describes the large $n$ behavior for probabilities of typical events, i.e., fluctuations of $\log Z_n(\theta)$ around $-2n\psi(\theta)$ of order $n^{1/3}$, it does not contain information about probabilities of rare events, like fluctuations of larger order. 
For instance, it does not contain any information about $\mathbb{P} \left(\frac{\log Z_n(\theta) + 2n \psi(\theta)}{(-\psi''(\theta)n)^{1/3}}\le r \right)$ as $r\to \pm\infty$ together with $n\to\infty$.
The upper tail asymptotics as $r\to +\infty$ are well understood \cite{GS}, but the lower tail ($r\to -\infty$) of the probability distribution for $\log Z_n(\theta)$ is not understood and cannot be analyzed with current techniques. 
Such lower tail probabilities are up to now only understood at zero temperature. Then, remarkably, $F_n^{\rm LPP}$ is for finite $n$ equal in distribution to the largest eigenvalue of an $n\times n$ Laguerre-Wishart random matrix \cite[Proposition 1.4]{Johansson}.
As a consequence, the distribution of $F_n^{\rm LPP}$ can be expressed as a Hankel determinant which can be analyzed in detail for large $n$, in order to obtain lower tail asymptotics for $\mathbb{P} \left(\frac{F_n^{\rm LPP} - 4 n  }{2^{4/3}n^{1/3}}\le r \right)$ as $r\to -\infty$ together with $n\to\infty$. In particular, we then have the large deviation result
\be\label{eq:largedeviationzerotemp}
\lim_{n\to\infty}\frac{-1}{n^2}\log\mathbb P\left[F_n^{\rm LPP}(\theta)\leq 4ns\right]=-\frac{s^2}{2}-\frac{3}{2}+2s - \log s,\qquad 0<s<1,
\ee
see e.g.\ \cite[Formula (162)]{DIK} or \cite{VivoMajumdarBohigas}. 
Observe that this is formally consistent with \eqref{eq: central limit theorem0temp}: setting $s=1+r(2n)^{-2/3}$ with $r\to -\infty$, we find from \eqref{eq:largedeviationzerotemp} that
$\log\mathbb{P} \left(\frac{F_n^{\rm LPP} - 4 n  }{2^{4/3}n^{1/3}}\le r \right)\sim -\frac{|r|^3}{12}$, which is indeed the leading order of $\log F_{\rm TW}(r)$ as $r\to -\infty$, see \cite[Theorem 1]{DIK}.
Lower tail deviations in the large temperature limit $\theta\to \infty$ can in principle be analyzed via the weak noise theory methods developed in \cite{KrajenbrinkLeDoussal2, Tsai2}. A large deviation principle and a variational problem for the stochastic six-vertex model, which degenerates to the log-Gamma polymer in an appropriate limit, was recently obtained in \cite{DLM}, but without explicit expression for the rate function.

\paragraph{Main result.}
Before formulating our main result, a conjecture for the large deviation rate function of the log-Gamma polymer partition function, we need to introduce several quantities. First, for $0<s<-\theta\psi(\theta)$ with $\psi=\Gamma'/\Gamma$ as before, we define $b=b(s,\theta)>0$ as the unique positive number solving the equation
\begin{align} \label{eq: defining equation for b}
\int_{0}^{1} \left(\theta\psi(\theta+iu\theta b/2)+\theta\psi(\theta-i u\theta b/2)+2s\right) \frac{du}{\pi\sqrt{1-u^2}}=0.
\end{align} 
Next, we define
\begin{align} \label{eq: def f}
f(s,\theta) &=  b^2 \int_{0}^{1}\left( \theta\psi(\theta+iu\theta b/2))+\theta\psi(\theta-iu\theta b/2)+2s\right) \sqrt{1-u^2}\frac{du}{2\pi},\qquad 0<s<-\theta\psi(\theta)
\end{align}
and \begin{equation}\label{def:F}F(s,\theta)=-\int_{s}^{-\theta\psi(\theta)} f(t;\theta)dt,\qquad \qquad 0<s<-\theta\psi(\theta).\end{equation}

\begin{conjecture}\label{conjecture:main}
There exists $\theta_0>0$ such that uniformly for $0<\theta<\theta_0$ and $\epsilon<s<-\theta\psi(\theta)$ for any $\epsilon>0$, we have
\begin{equation}\lim_{n\to\infty}\frac{-1}{n^2}\log\mathbb P\left[\log Z_n(\theta)\leq \frac{2n}{\theta}s\right]=F(s,\theta).\end{equation}
\end{conjecture}
\begin{remark}
We will prove in Proposition \ref{lemma: choice of b} that equation \eqref{eq: defining equation for b} has indeed a unique solution for $\theta>0, 0<s<-\theta\psi(\theta)$. Moreover, we will show that as $s\to -\theta\psi(\theta)$ with $\theta>0$ fixed, we have
\be\label{eq:bastau}
b(s,\theta)\sim 4\sqrt{\frac{s+\theta\psi(\theta)}{\theta^3\psi''(\theta)}}.
\ee
On the other hand, as $\theta\to 0$ with $s$ fixed, we obtain after a straightforward computation, using the fact that $\psi(z)\sim -1/z$ as $z\to 0$, that
\be\label{eq:bastheta}
b(s,0):=\lim_{\theta\to 0}b(s,\theta)= 2\sqrt{\frac{1}{s^2}-1}.
\ee
\end{remark}

\begin{remark}
From the $s\to -\theta\psi(\theta)$ behavior for $b(s,\theta)$, a careful Taylor expansion of $f(s,\theta)$ yields 
\be F(s,\theta)\sim \frac{2}{3\theta^3\psi''(\theta)}(s+\theta\psi(\theta))^3\qquad \mbox{as $s\to -\theta\psi(\theta)$.}\ee
Setting $s=-\theta\psi(\theta)+\frac{\theta}{2}(-\psi''(\theta))^{1/3}r n^{-2/3}$, we see that this is consistent with the cubic lower tail of the Tracy-Widom distribution in \eqref{eq: central limit theorem},  $\log F_{\rm TW}(r)\sim -\frac{|r|^3}{12}$ as $r\to -\infty$.
\end{remark}

\begin{remark}
Using the fact that $\psi(z)\sim -1/z$ as $z\to 0$ and the $\theta\to 0$ behavior for $b(s,\theta)$, we see that
\be f(s,0):=\lim_{\theta\to 0}f(s,\theta)=2+\frac{s}{4}b(s,0)^2-\sqrt{4+b(s,0)^2}=2-s-\frac{1}{s}.\ee
Consequently, 
\be
F(s,0):=\lim_{\theta\to 0}F(s,\theta)=-\frac{s^2}{2}-\frac{3}{2}+2s - \log s.
\ee
This provides an interesting consistency check: as it must, the above formula is precisely the large deviation rate function for last passage percolation with exponential weights given in \eqref{eq:largedeviationzerotemp}.
\end{remark}

\paragraph{Methodology.}

The starting point of our analysis is an alternative Fredholm determinant representation, equal to but of a different form than \eqref{eq:detid1}, which was established recently in \cite[Corollary 5.2]{CC24}. For $\theta\in(0,1)$, this result, with $N=n$, $a_k=0$, $\alpha_j=2\theta$, states that 
\begin{align} \label{eq:detid2}
	\mathbb{E}(e^{-e^{-s} Z_n(\theta)}) = \det(1 - \sigma_s L_{n}^\theta )_{L^2(\R)},
\end{align}
where 
\begin{align}\label{def:L0}
L_{n}^\theta(x,x')= \frac{1}{(2\pi i)^2} \int_{\mathcal C} du \int_{\delta+i\mathbb R} dv \frac{\Gamma(2\theta-v)^n\Gamma(u)^n}{\Gamma(2\theta-u)^n\Gamma(v)^n} \frac{e^{-vx+ux'}}{v-u},
\end{align}
with $\theta<\delta<\min\{2\theta,1\}$, $\mathcal C$ a positively oriented circle of radius $<\theta$ around $0$, and with
\be\label{def:w}
\sigma_s(x) = \frac{1}{1+e^{-x+s}}.
\ee
Equivalently,
\begin{align}\label{def:L}
L_{n}^\theta(x,x')= \frac{1}{(2\pi i)^2} \int_\Sigma du \int_{c+i\mathbb R} dv \frac{W(v)}{W(u)} \frac{e^{-vx+ux'}}{v-u},
\end{align}
with $W$ as in \eqref{def:W}, $c=\delta-\theta\in(0,\min\{\theta,1-\theta\})$, $\Sigma=-\theta+\mathcal C$.
In \eqref{eq:detid2}, the right hand side is simply a short-hand notation for the Fredholm series \eqref{def:Fredholm series}, here with kernel $K(x,x')=-\sigma_s(x)L_n^\theta(x,x')$, without requiring at this point that $K$ is the kernel of a trace-class operator.

\medskip

A useful feature of $L_n^\theta$ is that it is the correlation kernel of a signed biorthogonal measure $d\mu_n^\theta(x_1,\ldots, x_n)$ on $\mathbb R^n$ \cite{BorBiOE}, more specifically a signed polynomial ensemble of derivative type.
In other words, $L_n^\theta$ has the reproducing property
\be\int_{\mathbb R}L_n^\theta(x,t)L_n^\theta(t,x')dt=L_n^\theta(x,x'),\ee
and 
\be\label{def:BiOM}
d\mu_n^\theta(x_1,\ldots, x_n):=\frac{1}{n!}\det\left(L_n^\theta(x_i,x_j)\right)_{i,j=1}^n dx_1\cdots dx_n
\ee
has the form
\be\label{def:BiOM2}
d\mu_n^\theta(x_1,\ldots, x_n)=\frac{1}{Z_n}\left(\prod_{1\leq j<i\leq n} (x_j-x_i)\right)\ \det\left(G^{(j-1)}(x_i;\theta)\right)_{i,j=1}^n dx_1\cdots dx_n,
\ee
where $G(x;\theta)$ is given by
\be G(x;\theta)=
\frac{1}{2\pi i}\int_{{c+i\mathbb R}}\frac{\Gamma(\theta-v)^n}{\Gamma(\theta+v)^n}\frac{e^{-vx}}{(\theta+v)^{n}}d v,
\ee
see \cite[Section 5]{CC24}. Alternatively, $G$ can be expressed in terms of the Meijer $G$-function, but we will not use this fact.

\medskip

As we will point out, this biorthogonal measure lives on scales of order $n/\theta$ for large $n$ and small $\theta$, which leads us to defining the re-scaled kernel
\be\label{def:Lrescales}\widehat L_n^\theta(y,y'):=\frac{2n}{\theta}L_n^\theta\left(\frac{2n}{\theta}y,\frac{2n}{\theta}y'\right)=\frac{1}{(2\pi i)^2} \int_\Sigma du \int_{c+i\mathbb R} dv \frac{\Gamma\left(\theta-\frac{\theta v}{2n}\right)^n\Gamma\left(\theta+\frac{\theta u}{2n}\right)^n}{\Gamma\left(\theta+\frac{\theta v}{2n}\right)^n\Gamma\left(\theta-\frac{\theta u}{2n}\right)^n} \frac{e^{-vy+uy'}}{v-u}
,\ee
where $\Sigma$ is now a circle of radius $<2n$ around $-2n$, and $0<c<\min\left\{2n,2n\frac{1-\theta}{\theta}\right\}$.
Then, the limit $\theta\to 0$ makes sense, and 
\[\widehat L_n^0(y,y')=\frac{1}{(2\pi i)^2} \int_\Sigma du \int_{i\mathbb R} dv \frac{\left(1+\frac{v}{2n}\right)^n\left(1-\frac{u}{2n}\right)^n}{\left(1-\frac{ v}{2n}\right)^n\left(1+\frac{u}{2n}\right)^n} \frac{e^{-vy+uy'}}{v-u}\] is a well-known double integral expression for the kernel for the Laguerre-Wishart random matrix ensemble (see, e.g., \cite[Formula (4.2)]{DesrosiersForrester} or \cite[Formula (5.158)]{Forrester}), which satisfies the Marchenko-Pastur law on $[0,1]$:
\be\label{eq:MP}
\lim_{n\to\infty}\frac{1}{n}\widehat L_n^0(y,y)=\frac{2}{\pi}\sqrt{\frac{1-y}{y}}1_{(0,1)}(y),\qquad y\in\mathbb R. 
\ee

From \eqref{eq:detid2}, we obtain the identity
\begin{align} \label{eq:detid3}
	\mathbb{E}(e^{-e^{-\frac{2n}{\theta}s} Z_n(\theta)}) = \det(1 - {\sigma_{s,n,\theta}} \widehat L_{n}^\theta )_{L^2(\R)},\qquad \sigma_{s,n,\theta}(y)=\frac{1}{1+e^{-\frac{2n}{\theta}(y-s)}}.
\end{align}
As $n\to\infty$, $\sigma_{s,n,\theta}(y)$ converges to the step function $1_{(s,+\infty)}(y)$, such that one may expect that the logarithm of the {\em step function Fredholm determinant}
\be\label{def:Fredholmdetstep}
Q_n^\theta(s):=\det(1 - 1_{(s,+\infty)} \widehat L_{n}^\theta)_{L^2(\R)}=\det(1 - \widehat L_{n}^\theta)_{L^2(s,+\infty)}
\ee
is a good approximation of the logarithm of the {\em smoothed Fredholm determinant}
\be\label{def:Fredholmdetsmooth}
\widetilde Q_n^\theta(s):=\det(1 - {\sigma_{s,n,\theta}} \widehat L_{n}^\theta )_{L^2(\R)}=\det(1 - {\sigma_{\frac{2ns}{\theta}}} L_{n}^\theta)_{L^2(\R)},
\ee
for large $n$.
We can however not prove this fact without further assumptions, and this is the only reason why Conjecture \ref{conjecture:main} is not a Theorem. It leads us to the following ansatz.
\begin{ansatz}\label{ansatz1}
There exists $\theta_0>0$ such that
\be\label{eq:LDPcomparison}\log Q_n^\theta(s)\sim\log \widetilde Q_n^\theta(s),\qquad\mbox{as $n\to\infty$,}
\ee
uniformly for $\epsilon<s<-\theta\psi(\theta)$ and $0<\theta<\theta_0$, for any $\epsilon>0$.
\end{ansatz}
We will explain in Section \ref{section:ansatz} why we are confident that our ansatz holds true. In Section \ref{section:proofconjecture}, we will prove Conjecture \ref{conjecture:main} under Ansatz \ref{ansatz1}.

\medskip

The advantage of the step function Fredholm determinant $Q_n^\theta(s)$ from \eqref{def:Fredholmdetstep}, compared to \eqref{def:Fredholmdetsmooth}, is that it can be transformed to an {\em integrable} Fredholm determinant, which one can characterize in terms of a Riemann-Hilbert (RH) problem thanks to a method developed by Its, Izergin, Korepin, and Slavnov \cite{IIKS}.
We derive this RH characterization in Section \ref{section:diffid}.

\medskip

A major part of this paper is devoted to the large $n$ asymptotic analysis of this RH problem. 
In Section \ref{section:RH}, we initiate the RH analysis. For $s>-\theta\psi(\theta)+\epsilon$ with $\epsilon>0$, the asymptotic analysis of the RH problem is simple and only requires a suitable deformation of jump contours. For $\epsilon\leq s\leq -\theta\psi(\theta)+\epsilon$ however, we need to introduce a $g$-function and construct  local parametrices.
For $s<-\theta\psi(\theta)-\epsilon$, we need to construct a global parametrix and two local Airy parametrices. We do this in Section \ref{section:RH2}. For $-\theta\psi(\theta)-\epsilon\leq s\leq -\theta\psi(\theta)+\epsilon$, the two local Airy parametrices come together and merge to a single local parametrix that we can build using a model RH problem associated to the Painlev\'e II equation. This is the content of Section \ref{section:RH3}. This analysis is in essence a {\em closing of a gap} transition similar to the one studied by Baik, Deift, and Johansson in their study of the longest increasing subsequence of a random permutation \cite{BaikDeiftJohansson}, and later in \cite{BleherIts, ClaeysKuijlaars} in a random matrix context. 
The asymptotic analysis of the RH problem yields the following result, which we will prove in Section \ref{section:RHproof}.

\begin{theorem} \label{theorem:main}
There exists $\theta_0>0$ such that 
\begin{align*}
\log Q_n^\theta\left(s\right) &= -n^2 F(s,\theta) + O(n^{2/3}) 
\end{align*}
as $n \to +\infty$ uniformly for $\epsilon <s< -\theta\psi(\theta)$ and $0<\theta<\theta_0$, for any $\epsilon>0$. The function $F(s,\theta)$ is given by \eqref{eq: def f}--\eqref{def:F} with $b(s,\theta)>0$ the unique positive number solving the equation \eqref{eq: defining equation for b}.
Furthermore, the function $F(s,\theta)$ is positive for $0<s<-\theta\psi(\theta)$ and $0<\theta<\theta_0$.
\end{theorem}
\begin{remark}
As we will explain later, we could with some more effort prove the Tracy-Widom convergence
\be\label{eq:TWcvgc}
\lim_{n\to\infty}Q_n^\theta\left(-\theta\psi(\theta)+\frac{\theta}{2}(-\psi''(\theta))^{1/3}r n^{-2/3}\right)= F_{\rm TW}(r),
\ee
uniformly for $r\in(r_0,+\infty)$ for any $r_0\in\mathbb R$.
This is consistent with \eqref{eq: central limit theorem} and Ansatz \ref{ansatz1}. The error term $O(n^{2/3})$ in Theorem \ref{theorem:main} is not sharp: we believe that it can be improved to $O(1)$. It is however not clear whether the approximation \eqref{eq:LDPcomparison} is valid up to that order.
\end{remark}
Combining Theorem \ref{theorem:main} with Ansatz \ref{ansatz1}, we will prove Conjecture \ref{conjecture:main} in Section \ref{section:proofconjecture}.

\section{RH characterization of the step function Fredholm determinant}\label{section:diffid}

In this section, we work with a more general class of kernels $L_n$ than the ones arising in the context of the log-Gamma Fredholm determinant.
We let
\begin{align} \label{eq: general kernel}
{L}_{n}(x,x') = \frac{1}{(2\pi i)^2} \int_\Sigma du \int_{c+i\mathbb R} dv \frac{W_1(v)}{W_2(u)} \frac{e^{-vx+ux'}}{v-u},\qquad c\in\mathbb R,
\end{align}
where $c\in\mathbb R$, $\Sigma$ is a simple closed positively oriented curve in the half plane $\Re z<c$,
 $W_1:c+i\mathbb R \to  \C$ is smooth, bounded, and not identically $0$, and  $W_2:\Sigma\to\mathbb C\setminus\{0\}$ is smooth and bounded.

If $W_1=W_2=W$ given by \eqref{def:W} and $0<c<\min\{\theta,1-\theta\}$, we have that $L_n$ is equal to $L_n^\theta$. If we take $W_1(z)=W_2(z)=\widehat W(z)$, with
\be\label{def:Wrescaled} \widehat W(z)=\frac{\Gamma\left(\theta\left(1-\frac{z}{2n}\right)\right)^n}{\Gamma\left(\theta\left(1+\frac{z}{2n}\right)\right)^n},\ee
and $0<c<\min\left\{2n,2n\frac{1-\theta}{\theta}\right\}$,
then $L_n$ is equal to the re-scaled kernel $\widehat L_n^\theta$ from \eqref{def:Lrescales}.

The goal of this section is to relate the step function Fredholm determinant $\det\left(1-1_{(-\infty,s)}L_n\right)_{L^2(\mathbb R)}$ to a $2\times 2$ matrix-valued RH problem.
For our purposes, it would be enough to consider $W_1,W_2$ given by \eqref{def:W} or \eqref{def:Wrescaled}. However, the results and proofs in this section do not rely on the form of $W_1$ and $W_2$, and other choices of $W_1,W_2$ are relevant in other models, like the non-homogeneous log-Gamma polymer on a rectangular lattice, the O'Connell-Yor polymer, the mixed polymer, and more general biorthogonal ensembles of derivative type. We therefore believe that Proposition \ref{prop:diffid} below, for general $W_1,W_2$, is of independent interest, and will turn out useful to study other models.
This section follows similar lines as \cite[Section 2]{CGS17} and is based on a general method developed in \cite{BertolaCafasso} and built on the theory of integrable operators from \cite{IIKS, DIZ}.

\medskip

The RH problem of interest is the following, and depends on parameters $s\in\mathbb R$ as well as on the function $W$, but we will omit the latter dependence in our notation. We will write $Y(z)=Y(z;s)$ for the solution of this RH problem.
\subsubsection*{RH problem for $Y$}
\begin{itemize}
	\item[(Y1)] $Y:\C \setminus ((c+i\mathbb R) \cup \Sigma) \to \C^{2 \times 2}$ is analytic.
	\item[(Y2)] For $z\in (c+i\mathbb R) \cup \Sigma$, $Y$ satisfies the jump condition
	\begin{align*}
		Y_+(z) = Y_-(z)J(z),
	\end{align*}
	where
	\begin{align}
		J(z) = \begin{cases}
			\begin{pmatrix}
				1 & 0 \\ e^{-sz} W_1(z)  & 1
			\end{pmatrix} &\text{ for $z \in c+i\mathbb R $}, \\ 
			\begin{pmatrix}
				1 & -e^{sz} W_2(z)^{-1} \\ 0 & 1
			\end{pmatrix} &\text{ for $z \in \Sigma $}.
		\end{cases}
	\end{align}
	\item[(Y3)] There exists a $2\times 2$ matrix $Y_1=Y_1(s)$ depending on $s$ and also on $W_1,W_2$, such that  $Y(z) = I + \frac{Y_1}{z} + \mathcal{O}(z^{-2})$ as $z \to \infty$.
\end{itemize}
In general, the boundary values $Y_\pm$ have to be understood in $L^2$-sense and condition (Y3) is valid for $z\to\infty$ away from $c+i\mathbb R$, like in \cite{DIZ}. However, if $W_1(c+iy)$ decays sufficiently fast as $y\to \pm\infty$, say $W_1(c+iy)=O(|y|^{-1-\epsilon})$ as $y\to\pm\infty$ for some $\epsilon>0$, then the boundary values are continuous, and condition (Y3) is uniform for $z\in\mathbb C\setminus (c+i\mathbb R)$. To avoid technical complications, we will henceforth assume that $W_1(c+iy)=O(|y|^{-1-\epsilon})$ as $y\to\pm\infty$. Note that in our case of interest \eqref{def:Wrescaled}, it follows from Stirling's approximation that 
$W_1(c+iy)=O(|y|^{-c})$ as $y\to \pm\infty$, which means that we need to take $c>1+\epsilon$, in addition to the already standing condition $c<\min\{2n,2n\frac{1-\theta}{\theta}\}$.

\medskip

The rest of this section is dedicated to the proof of the following result, connecting the Fredholm determinant
$\det(1 - 1_{(s,\infty)}L_n)_{L^2(\R)}$ with the solution to the above RH problem.
\begin{proposition}[Differential identity] \label{prop:diffid}
The RH problem for $Y$ is solvable if and only if $\det(1 - L_n)_{L^2(s,+\infty)}\neq 0$, and we have the identity
	\begin{align*}
		\frac{d}{ds} \log \det(1 - L_n)_{L^2(s,+\infty)} = \left(Y_1(s)\right)_{11}.
	\end{align*}
In particular,
$Q_n^\theta(s)$ from \eqref{def:Fredholmdetstep} satisfies
	\begin{align}\label{eq:diffidQ}
		\frac{d}{ds} \log Q_n^\theta(s) = \left(Y_1(s)\right)_{11},
	\end{align}
where $Y_1(s)$ corresponds to the RH solution $Y$ with $W_1(z)=W_2(z)=\widehat W(z)$ given by \eqref{def:Wrescaled}.
\end{proposition}

To prove this result, we will use the two-sided Laplace transform $\mathcal B:L^2(\mathbb R)\to L^2(i\mathbb R)$ defined by
\begin{align*}
\mathcal{B} [f](v) = \int_{-\infty}^{\infty} e^{-vx} f(x) dx,
\end{align*}
with inverse $\mathcal B^{-1}:L^2(i\mathbb R)\to L^2(\mathbb R)$ given by
\begin{align*}
\mathcal{B}^{-1} [\phi] (x) = \frac{1}{2 \pi i} \int_{i \R} e^{vx} \phi(v) dv,
\end{align*}
as well as the integral operator $H_n^s :L^2( i\mathbb R) \to L^2( i\mathbb R)$  with kernel
\begin{align*}
H_n^s (v,v') =\frac{1}{(2 \pi i)^2} \int_{\mathcal C} du e^{s(u-v)} \frac{W_1(c+v')}{W_2(c+u) (v'-u) (v-u)},
\end{align*}
where $\mathcal C=\Sigma-c$.
Observe that $H_n^s$ is a trace-class operator, as $H_n^s=J\circ K$ is the composition of the Hilbert-Schmidt operators $J:L^2(\mathcal C) \to L^2(i\mathbb R),K:L^2( i\mathbb R) \to L^2( \mathcal C)$ with kernels \[J(v,u)=\frac{e^{-sv}}{2\pi iW_2(c+u)(v-u)},\qquad K(u,v')=\frac{e^{su}W_1(c+v')}{2\pi i(v'-u)}.\]

\begin{lemma} \label{lemma: det L = det H}
For any $s\in\mathbb R$, we have the operator identity
\[\mathcal B^{-1}\circ H_n^s\circ \mathcal B=\mathcal L_n^s,\]
where $\mathcal L_n^s:L^2(\mathbb R)\to L^2(\mathbb R)$ is defined by
\[\left(\mathcal L_n^s\right)[f](x)=1_{(s,+\infty)}(x)e^{-cx}\int_{\mathbb R}L_n(x',x)e^{cx'}f(x')dx',\] and the Fredholm determinant identity
\begin{align*}
\det(1 - L_n)_{L^2(s,+\infty)} = \det(I - H_n^s )_{L^2(i\mathbb R)}.
\end{align*}
\end{lemma}
\begin{proof}
Let us first compute the kernel of the operator $\mathcal B^{-1}\circ J:L^2(\mathcal C)\to L^2(\mathbb R)$: we have
\begin{align*}
(\mathcal B^{-1}\circ J)(x,u)&=\frac{1}{(2\pi i)^2W_2(u+c)}\int_{i\mathbb R}\frac{e^{(x-s)v}}{v-u}dv\\
&=\frac{e^{(x-s)u}}{2\pi i W_2(u+c)}1_{(s,+\infty)}(x),
\end{align*}
where we used a contour deformation and residue argument for the last identity.
Next, we compute the kernel of the operator $K\circ \mathcal B:L^2(\mathbb R)\to L^2(\mathcal C)$: we have
\begin{align*}
(K\circ \mathcal B)(u,x)&=\frac{e^{su}}{2\pi i}\int_{i\mathbb R}\frac{e^{-xv}W_1(v+c)}{v-u}dv.
\end{align*}
Combining these expressions, we find that the kernel of $\mathcal B^{-1}\circ H_n^s\circ \mathcal B=\mathcal B^{-1}\circ J\circ K\circ \mathcal B$
 is given by
\begin{align*}
(\mathcal B^{-1}\circ H_n^s\circ \mathcal B)(x,x')&=1_{(s,+\infty)}(x)\frac{1}{(2\pi i)^2}\int_{\mathcal C}\frac{e^{xu}}{W_2(u+c)}\int_{i\mathbb R}\frac{e^{-x'v}W_1(v+c)}{v-u}dv\ du\\
&=
1_{(s,+\infty)}(x)\frac{e^{-c(x-x')}}{(2\pi i)^2}\int_{\Sigma} du\int_{c+i\mathbb R}dv\frac{W_1(v)}{W_2(u)}\frac{e^{xu-x'v}}{v-u}dv\ du\\
&=1_{(s,+\infty)}(x)e^{-c(x-x')}L_n(x',x)
.
\end{align*}
It follows that the kernel $1_{(s,+\infty)}(x)e^{-c(x-x')}L_n(x',x)$ defines a trace-class operator $\mathcal L_n^s$ (with $W_1=W_2=\widehat W$ as in \eqref{def:Wrescaled}, it is even of finite rank $n$, see \cite{CC24}, but we don't need this). We thus have the Fredholm determinant identities 
\begin{align*}
\det(I - H_n^s)_{L^2( i\mathbb R)}&=\det(I - \mathcal B^{-1}\circ H_n^s\circ \mathcal B )_{L^2( \mathbb R)}\\
&=\det(1 - \mathcal L_n^s)_{L^2(\R)}.\end{align*}
By \eqref{def:Fredholm series}, it is straightforward to check that the latter is precisely the Fredholm series $\det(1 - L_n)_{L^2(s,+\infty)}$, and the result is proven.
\end{proof}
Next, we can write the Fredholm determinant in a different way, in terms of a Fredholm determinant of an integrable kernel operator.
\begin{lemma} \label{lemma: det H = det M}
Let $s \in \R$. Then
\begin{align*}
 \det(I - H_n^s )_{L^2( i\mathbb R)} =  \det(I - M_n^s )_{L^2((c+i\mathbb R) \cup \Sigma)},
\end{align*}
where $M_n^s: L^2((c+i\mathbb R) \cup \Sigma) \to L^2( (c+i\mathbb R) \cup \Sigma)$ is the integral operator induced by the kernel 
\begin{align}\label{def:Mns}
M_n^s(z,z') = \frac{\textbf{f}(z)^t \textbf{g}(z')}{z-z'}
\end{align}
with
\begin{align}\label{def:Mnsinteg}
\textbf{f}(z) = \frac{1}{2 \pi i} \begin{pmatrix}
 1_{\Sigma}(z)W_2(z)^{-1} \\  1_{c+ i\mathbb R}(z) W_1(z)
\end{pmatrix}, \qquad \textbf{g}(z)=\begin{pmatrix}
    -e^{-sz}1_{c+i\mathbb R}(z) \\ e^{sz}1_{\Sigma}(z)
\end{pmatrix}.
\end{align}
\end{lemma}

\begin{proof}
We have
\[\widetilde H_n^s(v,v'):=H_n^s(v-c,v'-c)=\frac{1}{(2 \pi i)^2} \int_{\Sigma} du e^{s(u-v)} \frac{W_1(v')}{W_2(u) (v'-u) (v-u)}.\]
Hence, similarly to the factorization $H_n^s=J\circ K$ used before, we have
\begin{align*}
\widetilde H_n^s = A \circ B,
\end{align*}
where $A:L^2(\Sigma) \to L^2(c+i\mathbb R)$ and  $B:L^2(c+ i\mathbb R) \to L^2(\Sigma)$ are the integral operators induced by the kernels
\begin{align*}
A(v,u) = \frac{e^{-sv}}{2 \pi i W_2(u)(v-u)}, \qquad B(u,v') = \frac{e^{su}W_1(v')}{2 \pi i (v'-u)  }.
\end{align*}
Then $A$ and $B$ are Hilbert-Schmidt operators since $\int_{\Sigma} \int_{c+ i\mathbb R} |A(v,u)|^2 |dv||du| < \infty$ and \\ $\int_{c+ i\mathbb R} \int_{\Sigma} |B(u,v')|^2 |du||dv'| < \infty$. 
{We can decompose further $B=B_2 \circ B_1$ and $A = A_2 \circ A_1$, where the operators $B_i,A_i$ are given by
\begin{align*}
	&B_1: L^2(c+ i\mathbb R) \to L^2\left( \gamma\right), &\qquad B_1(w,v') &= \frac{W_1(v')}{2\pi i(v'-w)}, \\
	&B_2: L^2\left(\gamma \right) \to L^2(\Sigma), &\qquad B_2(u,w) &= \frac{e^{su}}{2\pi i(w-u)}, \\
	&A_1: L^2(\Sigma) \to L^2\left(\gamma\right), &\qquad A_1(w,u) &= \frac{1}{2\pi iW_2(u)(w-u)}, \\
	&A_2: L^2\left(\gamma \right) \to L^2(c+ i\mathbb R), &\qquad A_2(v,w) &= \frac{e^{-sv}}{2\pi i(v-w)}.
\end{align*}
Here $\gamma$ is a closed contour surrounding $\Sigma$ in the half plane $\Re z<c$. All the operators above are Hilbert-Schmidt, hence $A$ and $B$ are trace-class.}
Considering $M_n^s$ as an operator acting on the space $L^2(\Sigma) \oplus L^2(c+ i\mathbb R)$ and writing it as a $2\times 2$ matrix of operators
\begin{align*}
	M_n^s = \begin{pmatrix}
		0 & A \\ B & 0
	\end{pmatrix},
\end{align*}
we can compute:
\begin{align*}
	\det (1 -H_n^s)_{L^2(c+i \mathbb R)} &= \det \left(  1 - \begin{pmatrix}
		A \circ B & 0 \\ 0 & 0
	\end{pmatrix} \right)_{L^2((c+i \mathbb R)\cup \Sigma)} \\
	&= \det \left(  1 - \begin{pmatrix}
		A \circ B & 0 \\ B & 0
	\end{pmatrix} \right)_{L^2((c+i \mathbb R)\cup\Sigma)}\\
	&=\det \left(  1 + \begin{pmatrix}
		0 & A \\ 0 & 0
	\end{pmatrix} \right)_{L^2((c+i \mathbb R)\cup\Sigma)}\det \left(  1 - \begin{pmatrix}
		0 & A \\ B & 0
	\end{pmatrix} \right)_{L^2((c+i \mathbb R)\cup\Sigma)} \\
	&=\det \left(  1 - \begin{pmatrix}
		0 & A \\ B & 0
	\end{pmatrix} \right)_{L^2((c+i \mathbb R)\cup\Sigma)} = \det (1 -M_n^s)_{L^2((c+i \mathbb R)\cup\Sigma)} ,
\end{align*}
which proves the lemma.
\end{proof}
Following the terminology of \cite{IIKS}, an integral operator with kernel of the form \eqref{def:Mns} with $\mathbf f(u)^t\mathbf g(u)=0$ is called integrable. The general theory behind such operators developed in \cite{IIKS, DIZ} and \cite{BertolaCafasso} enables us to relate the Fredholm determinant $\det(I - M_{n,s} )_{L^2(i \mathbb R\cup\Sigma)}$ to the solution of a RH problem with jump matrix $I-2\pi i {\mathbf f}(z){\mathbf g}(z)^t$. In our case of interest, the relevant RH problem is given by conditions (Y1)--(Y3). Here and below, we write $\sigma_3=\begin{pmatrix}1&0\\0&-1\end{pmatrix}$.

\begin{lemma} \label{lemma: auxiliary lemma diff identity}
Let $s\in\mathbb R$. It holds that
\begin{align*}
\frac{d}{ds} \log \det(I - M_{n}^s )_{L^2(i \mathbb R\cup\Sigma)} 
&=  \int_{(c+ i\mathbb R)\cup \Sigma} \frac{z}{2} \tr\left( Y_+(z)^{-1}Y_+'(z)\sigma_3 -  Y_-(z)^{-1}Y_-'(z)\sigma_3 \right) \frac{dz}{2 \pi i}.
\end{align*}
\end{lemma}

\begin{proof}
By \cite[Theorem 2.1]{BertolaCafasso}, we have
\begin{align*}
\frac{d}{ds} \log \det(I - M_{n}^s )_{L^2((c+i \mathbb R)\cup\Sigma)} &= \int_{(c+ i\mathbb R)\cup \Sigma} z \tr\left( Y_-(z)^{-1}Y_-'(z) \partial_sJ(z)J(z)^{-1} \right)  \frac{dz}{2 \pi i}.
\end{align*}
Using that $(I-J)J^{-1} = I-J$ and abbreviating $J=J(z)$, we get
\begin{align*}
\partial_sJ(z) J(z)^{-1} = \begin{cases}
-z \begin{pmatrix}
0 &  0 \\ e^{-sz} W_1(z) & 0
\end{pmatrix}J^{-1}= z(I-J)J^{-1}=z(I-J) = z(J-I) \sigma_3, &\text{ for $z \in c+ i\mathbb R $}, \\ 
z\begin{pmatrix} 
0 & -e^{sz} W_2(z)^{-1}  \\ 0 & 0
\end{pmatrix}J^{-1} = z(J-I)J^{-1}=z(J-I) =z(J-I) \sigma_3, &\text{ for $z \in \Sigma $},
\end{cases}
\end{align*}
which yields \[
\frac{d}{ds} \log \det(I - M_{n}^s )_{L^2((c+i \mathbb R)\cup\Sigma)}
= \int_{(c+i \mathbb R)\cup \Sigma} z \tr\left( Y_-(z)^{-1}Y_-'(z)(J-I)\sigma_3 \right) \frac{dz}{2 \pi i}.
\] Furthermore, we have
\begin{align*}
Y_+(z)^{-1}Y_+'(z) &= J(z)^{-1}Y_-(z)^{-1}(Y_-'(z)J(z) + Y_-(z)J'(z)) \\
&= J(z)^{-1}Y_-(z)^{-1}Y_-'(z)J(z) +J(z)^{-1}J'(z).
\end{align*}
Since $J'(z)$ is off-diagonal this implies
\begin{align*}
\tr(Y_+(z)^{-1}Y_+'(z) \sigma_3) = \tr( Y_-(z)^{-1}Y_-'(z)J(z) \sigma_3J(z)^{-1}).
\end{align*}
Using $J^{-1}(z) = 2I-J(z)$ and $J(z)\sigma_3 = \sigma_3 J(z)^{-1}$, we obtain
\begin{align*}
\tr(Y_+(z)^{-1}Y_+'(z) \sigma_3) &= \tr( Y_-(z)^{-1}Y_-'(z)J(z) \sigma_3(2I-J(z))) \\
&= 2\tr( Y_-(z)^{-1}Y_-'(z)J(z) \sigma_3) - \tr( Y_-(z)^{-1}Y_-'(z)J(z) \sigma_3J(z)) \\
&=  2\tr( Y_-(z)^{-1}Y_-'(z)J(z) \sigma_3) - \tr( Y_-(z)^{-1}Y_-'(z)\sigma_3).
\end{align*}
Hence
\begin{align*}
&\tr(Y_+(z)^{-1}Y_+'(z) \sigma_3 -  Y_-(z)^{-1}Y_-'(z) \sigma_3) \\&= 2\tr( Y_-(z)^{-1}Y_-'(z)J(z) \sigma_3)  - 2\tr( Y_-(z)^{-1}Y_-'(z)\sigma_3) \\& = 2\tr\left( Y_-(z)^{-1}Y_-'(z)(J-I)\sigma_3 \right),
\end{align*}
which implies the result.
\end{proof}

We are now ready to prove Proposition \ref{prop:diffid}. 
By Lemma \ref{lemma: auxiliary lemma diff identity} we can write
\begin{align*}
	\frac{d}{ds} \log \det(I - M_{n}^s)_{L^2((c+i \mathbb R)\cup\Sigma)} =& \int_{c+ i\mathbb R} \frac{z}{2} \tr\left( Y_+(z)^{-1}Y_+'(z)\sigma_3 \right) \frac{dz}{2 \pi i} 
 -  \int_{c+ i\mathbb R} \frac{z}{2} \tr\left( Y_-(z)^{-1}Y_-'(z)\sigma_3 \right) \frac{dz}{2 \pi i} 
  \\
& + \int_{\Sigma} \frac{z}{2} \tr\left( Y_+(z)^{-1}Y_+'(z)\sigma_3 \right) \frac{dz}{2 \pi i} 
   -  \int_{\Sigma} \frac{z}{2} \tr\left( Y_-(z)^{-1}Y_-'(z)\sigma_3 \right) \frac{dz}{2 \pi i}.
\end{align*}

The integral
\begin{align*}
 \int_{\Sigma} \frac{z}{2} \tr\left( Y_+(z)^{-1}Y_+'(z)\sigma_3 \right) \frac{dz}{2 \pi i} 
\end{align*}
vanishes since the integrand can be continued analytically inside $\Sigma$. The contour of the integral 
\begin{align*}
 -  \int_{\Sigma} \frac{z}{2} \tr\left( Y_-(z)^{-1}Y_-'(z)\sigma_3 \right) \frac{dz}{2 \pi i} 
\end{align*}
can be enlarged because the integrand can be continued analytically to the region outside $\Sigma$ in the half plane $\Re z<c$. For instance, for $R>0$ large and $\epsilon>0$ small, we can write
\begin{align*}
 -  \int_{\Sigma} \frac{z}{2} \tr\left( Y_-(z)^{-1}Y_-'(z)\sigma_3 \right) \frac{dz}{2 \pi i}  = -  \int_{\Sigma_{\epsilon,R}} \frac{z}{2} \tr\left( Y(z)^{-1}Y'(z)\sigma_3 \right) \frac{dz}{2 \pi i},
\end{align*}
where $\Sigma_{\epsilon,R}$ consists of the semi-circle of radius $R$ around $c-\epsilon$ lying on the left of $ c-\epsilon+i\mathbb R$ and of the interval $[c-\epsilon-i R,c-\epsilon +i R]$.
As $\epsilon\to 0$ and $R\to \infty$, the integral on 
$[c-\epsilon-i R,c-\epsilon +i R]$ cancels out against
\begin{align*}
\int_{c+ i\mathbb R} \frac{z}{2} \tr\left( Y_+(z)^{-1}Y_+'(z)\sigma_3 \right) \frac{dz}{2 \pi i}. 
\end{align*}
Moreover, we can deform the contour $[c-iR,c+iR]$ in \begin{align*}
\int_{c-iR}^{c+iR} \frac{z}{2} \tr\left( Y_-(z)^{-1}Y_-'(z)\sigma_3 \right) \frac{dz}{2 \pi i}
\end{align*}
by analytic continuation of the integrand to the semi-circle of radius $R$ around $c$ on the right of $c$.
Hence 
\begin{align*}
	\frac{d}{ds} \log \det(I - M_{n}^s)_{L^2((c+i \mathbb R)\cup\Sigma)} =& \lim_{R\to\infty} \int_{\{ |z-c| =R \}} \frac{z}{2} \tr\left( Y(z)^{-1}Y'(z)\sigma_3 \right) \frac{dz}{2 \pi i},
\end{align*}
where the circle $\{ |z-c| =R \}$ is oriented clockwise. Using that $Y(z) = I + \frac{Y_1(s)}{z} + \mathcal{O}(z^{-2})$ we arrive at
\begin{align*}
\frac{d}{ds} \log \det(I - M_{n}^s)_{L^2((c+i \mathbb R)\cup\Sigma)} =& \int_{\{ |z-c| =R \}}  -\frac{Y_1(s)\sigma_3}{2 z}\frac{dz}{2 \pi i} = (Y_1(s))_{11},
\end{align*}
where we used that $Y_1(s)$ has zero trace because $\det Y=1$.

We now obtain Proposition \ref{prop:diffid} by combining this with Lemma \ref{lemma: det L = det H} and Lemma \ref{lemma: det H = det M}.

\section{Steepest descent analysis}\label{section:RH}
In this section we carry out the first steps in the nonlinear steepest descent analysis for the RH problem (Y1)--(Y3) in the case of the log-Gamma polymer. This means that $W_1$ and $W_2$ in conditions (Y1)--(Y3) are given by  
\begin{align} \label{eq: def F log Gamma rescaled}
	W_1(z)=W_2(z)=\widehat W(z)= \frac{ \Gamma(\theta(1-\frac{z}{2n}))^n}{  \Gamma(\theta(1+\frac{z}{2n}))^n}.
\end{align}
The associated kernel is given by
\begin{align}\label{eq: def kernel loggamma rescaled}
	\widehat{L}_{n}^\theta(x,x')= \frac{1}{(2\pi i)^2} \int_\Sigma du \int_{c+i\mathbb R} dv \frac{\widehat W(v)}{\widehat W(u)} \frac{e^{-xv+x' u}}{v-u}.
\end{align}
As explained before, we need $1+\epsilon<c<\min\{2n,2n\frac{1-\theta}{\theta}\}$, and $\Sigma$ lies in the left half plane and encircles $-2n$ in the positive direction, but none of the other zeros of $\widehat W$, namely $-2n-\frac{2n}{\theta}, -2n-\frac{4n}{\theta} \ldots$.
A convenient choice for $c$ will turn out to be $c=\delta n$ with $\delta>0$ sufficiently small.

Following the Deift-Zhou nonlinear steepest descent method \cite{DZ, Deift}, we will apply a series of transformations 
\[Y\mapsto U\mapsto T\mapsto S\mapsto R\] to the RH problem for $Y$, with the aim of arriving at a RH problem with small jumps as $n\to\infty$ and normalized to $I$ at infinity.
All of these RH solutions $Y,U,T,S,R$ will depend on a complex variable denoted $z$ or $\zeta$, as well as on parameters $s,\theta,n$. For the ease of notation, the dependence on the parameters $s,\theta,n$ will however not always be visible in our notations.

\subsection{Transformation $Y \mapsto U$}
Let us first introduce a rotation and dilation of the variable $z$ of the RH problem, by setting $z=in\zeta$. This is convenient for two reasons: first,  the scale of the new variable $\zeta$ will be convenient for the construction of a $g$-function, and secondly, we prefer to have the jump matrix on a horizontal line instead of a vertical one.
We define
\begin{align*}
	U(\zeta) = Y(in\zeta).
\end{align*} 
The RH conditions (Y1)--(Y3) for $Y$ then translate to the following RH conditions for $U$.

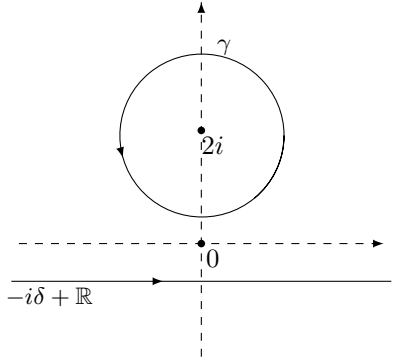
\begin{figure}[H]\label{figure:U}
	\begin{center}
		\begin{tikzpicture}		
			\node at (0,0) {};
			\fill (0,0) circle (0.05cm);
			\node at (0.15,-0.2) {$0$};
			\fill (0,1.5) circle (0.05cm);
			\node at (0.15,1.3) {$2i$};
			\draw[dashed,->-=1,black] (0,-1.5) to [out=90, in=-90] (0,3.2);
			\draw[dashed,->-=1,black] (-2.4,0) to [out=0, in=-180] (2.4,0);
			\node at (0.3,2.6) {\small $\gamma$};
					\draw[-<-=0.6,black] (2.5,-0.5)--($(2.5,-0.5)+(180+0:5)$);
		\node at (-2,-0.7) {\small $-i\delta+\mathbb R$};	
		\draw[->-=0.6,black] ([shift=(0.180:1.0cm)]-0.3,0.6) arc (-50.180:360:1.08cm);
\end{tikzpicture}
		\caption{The shape of the jump contour $(-i\delta+\mathbb R)\cup\gamma$ for the RH problem for $U$.}
	\end{center}
\end{figure}

\subsubsection*{RH problem for $U$}
\begin{itemize}
	\item[(U1)] $U:\C \setminus ((-i\delta+\mathbb R) \cup \gamma) \to \C^{2 \times 2}$ is analytic, where $\gamma=-\frac{i}{n}\Sigma$ is a closed positively oriented curve in the upper half plane surrounding $2i$ but not surrounding $2i+\frac{2i}{\theta}, 2i+\frac{4i}{\theta}, \ldots$.
	\item[(U2)] For $\zeta\in (-i\delta+\mathbb R) \cup \gamma$, $U$ satisfies the jump condition
	\begin{align*}
		U_+(\zeta) = U_-(\zeta)J_U(\zeta),
	\end{align*}
	where
	\begin{align}
		J_U(\zeta) = \begin{cases}
			\begin{pmatrix}
				1 & 0 \\ e^{nh(\zeta;s,\theta)} & 1
			\end{pmatrix} &\text{ for $\zeta \in -i\delta +\mathbb R $}, \\ 
			\begin{pmatrix}
				1 & -e^{-nh(\zeta;s,\theta)} \\ 0 & 1
			\end{pmatrix} &\text{ for $\zeta \in \gamma $},
		\end{cases}
\label{def:JU}\end{align}		
		and
\be\label{def:h}
		h(\zeta;s,\theta)=\log\Gamma\left(\theta\left(1-\frac{i\zeta}{2}\right)\right)-\log\Gamma\left(\theta\left(1+\frac{i\zeta}{2}\right)\right)-is\zeta.\ee
	\item[(U3)] As $\zeta\to\infty$, we have $U(\zeta) = I -i \frac{Y_1(s,n,\theta)}{n\zeta} + \mathcal{O}(\zeta^{-2})$, where we write $Y_1(s,n,\theta)$ instead of $Y_1(s)$ to emphasize the dependence on $n$ and $\theta$.
\end{itemize}
Observe that $e^{\pm nh(\zeta;s,\theta)}$ is analytic for $\zeta\in\mathbb C\setminus\{\mp 2 i,\mp 2 i\mp\frac{2i}{\theta}, \mp 2i\mp\frac{4i}{\theta},\cdots\}$, and that
\be\label{eq:derh}
		h'(\zeta;s,\theta)=-\frac{i\theta}{2}\psi\left(\theta\left(1-\frac{i\zeta}{2}\right)\right)-\frac{i\theta}{2}\psi\left(\theta\left(1+\frac{i\zeta}{2}\right)\right)-is,\ee
and
\be\label{eq:der2h}
		h''(\zeta;s,\theta)=-\frac{\theta^2}{4}\psi'\left(\theta\left(1-\frac{i\zeta}{2}\right)\right)+\frac{\theta^2}{4}\psi'\left(\theta\left(1+\frac{i\zeta}{2}\right)\right).\ee
In particular, $h''(\zeta;s,\theta)$ is antisymmetric in $\zeta$, and for $\zeta>0$,
we have that
\be\label{eq:der2hreal}
		ih''(\zeta;s,\theta)=-\frac{\theta^2}{2}\Im\psi'\left(\theta\left(1+\frac{i\zeta}{2}\right)\right)>0.\ee
The positivity follows for instance from the classical identity
$\psi'(z)=\sum_{k=0}^\infty\frac{1}{(z+k)^2}$
for the trigamma function.

As $\theta\to 0$, we have
\be\label{eq:limith}
h(\zeta;s,0):=\lim_{\theta\to 0}h(\zeta;s,\theta)=\log\frac{\zeta-2i}{\zeta+2i}
-is\zeta.\ee
The sign of $\Re h(\zeta;s,0)$ as a function of $\zeta\in\mathbb C$ will be crucial in what follows.
We denote $\mathcal H_+^s$ for the subset of the upper half plane containing all $\zeta$ such that $\Re h(\zeta;s,0)>0$, and $\mathcal H_-^s=\overline{\mathcal H_+^s}=-\mathcal H_+^s$ for the subset of the lower half plane on which $\Re h(\zeta;s,0)<0$.
 We illustrate the sign distribution of $\Re h(\zeta;s,0)$ for $s<1$, for $s=1$, and for $s>1$ in Figure \ref{figure:zeroset}.

\begin{figure}[htp]
\centering
\includegraphics[width=.32\textwidth]{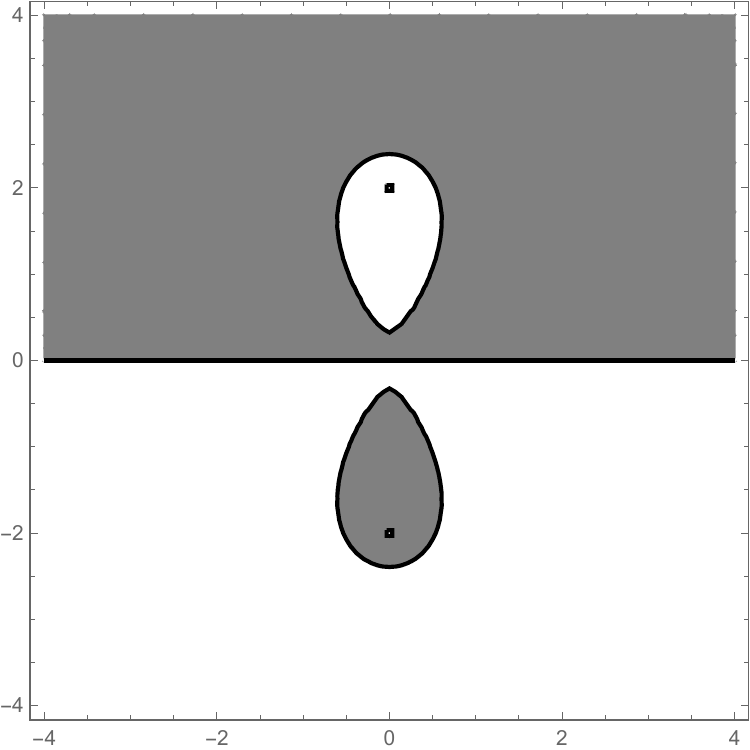}\hfill
\includegraphics[width=.32\textwidth]{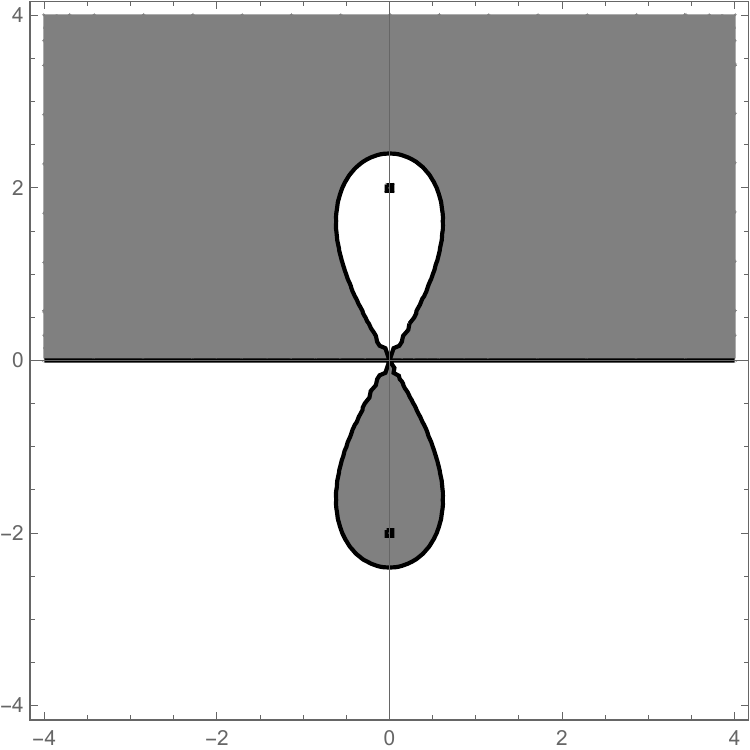}\hfill
\includegraphics[width=.32\textwidth]{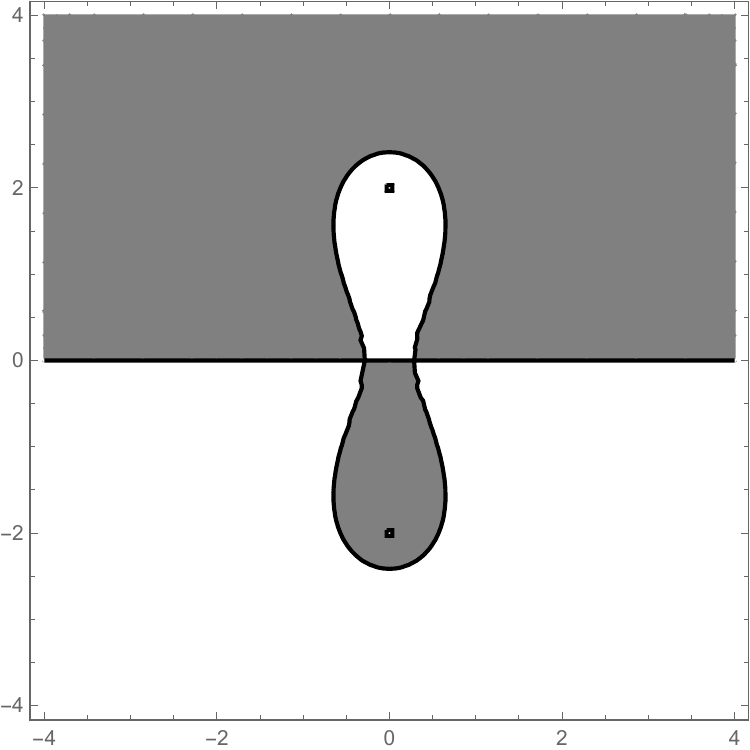}
\caption{Plot of regions in the complex $\zeta$-plane where $\Re h(\zeta;s,0)$ is positive (shaded) and negative (white). The left figure corresponds to $s>1$, the middle figure to $s=1$, and the right figure to $s<1$. $\mathcal H_+^s$ is the shaded region in the upper half plane, and $\mathcal H_-^s$ is the white region in the lower half plane.
}
\label{figure:zeroset}
\end{figure}

We record the following result about the sets $\mathcal H_\pm^s$ for $s=1$ and $s>1$ for later use.
\begin{proposition}\label{prop:signh}
The sets $\mathcal H_\pm^s$ are increasing in $s$: for any $s_1<s_2$, we have $\mathcal H_\pm^{s_1}\subset \mathcal H_\pm^{s_2}$.
For $s=1$, there exists a smooth $8$-shaped closed curve with self-intersection at $0$, such that the exterior of the curve is a subset of $\mathcal H_+^1\cup\mathcal H_-^1$, and such that angle between the curve and the real line at $0$ is $\pm\pi/3$.
For any $s_0>1$, there exist two complex conjugate simply closed curves $\tilde{\gamma}_\pm$, with $\tilde\gamma_+$ strictly in the upper half plane and $\tilde\gamma_-$ strictly in the lower half plane, such that the exterior of $\tilde\gamma_+\cup\tilde\gamma_-$ is a subset of $\mathcal H_+^s\cup \mathcal H_-^s$ for every $s>s_0$.
\end{proposition}
\begin{proof}
The fact that $\mathcal H_\pm^s$ is increasing in $s$ is a simple consequence of \eqref{eq:limith}.

For $s=1$, we compute
\[\Re h(\zeta;1,0)=\log\frac{|\zeta-2i|}{|\zeta+2i|}+\Im\zeta.\]
As $\zeta\to\infty$, this behaves like $\Im\zeta+O(\zeta^{-2})$, and it is straightforward to check that there exists $R>0$ such that $\pm \Re h(\zeta;1,0)>0$ for $|\zeta|\geq R$ and $\pm\Im\gamma>0$.
Next, given $R>0$, we prove that there exists $t_0=t_0(R)>0$ such that 
$\pm \Re h(\zeta;1,0)>0$ for $|\zeta|< R$, provided that $|\arg\zeta|<t_0$ or $|\arg(-\zeta)|<t_0$, and $\pm\Im\zeta>0$.
For that, we observe that 
\[h'(\zeta;1,0)=\frac{4i}{\zeta^2+4}-i,\]
such that for $u\in \mathbb R$, $u\neq 0$,
\[\Re h(u;1,0)=0,\qquad\partial_v\left. \Re h(u+iv;1,0)\right|_{v=0}=-\partial_u\left. \Im h(u+iv;1,0)\right|_{v=0}>0,\]
by the Cauchy-Riemann conditions. Hence, for $t_0>0$ sufficiently small, we have the required inequalities.
With the values of $R, t_0>0$ at hand, we can now choose the eight-shaped in the region $\{|\zeta|\geq R\}\cup\{|\zeta|<R,\mbox {and } |\arg\zeta|<t_0\mbox{ or }|\arg(-\zeta)|<t_0\}$. From the cubic behavior of $h(\zeta;1,0)$ around $0$, it is clear that we can take  the curve such that it makes an angle $\pm\pi/3$ with the real line.

The case $s>1$ can be treated similarly, the only difference being that a region containing $0$ can be connected with the exterior of a large disk, such that we can choose $\tilde\gamma_+$ in a region of the form $\{|\zeta|\geq R\}\cup\{|\zeta|<R\mbox { and } 0<\Im\zeta<t_0\}$.
\end{proof}

\subsection{Transformation $U \mapsto T$}

The phase diagrams in Figure \ref{figure:zeroset}
indicate how we need to proceed with the asymptotic RH analysis for $\theta>0$ sufficiently small. In an ideal scenario, we would like to have a RH problem with jump matrices close to identity. To achieve this, we need to  have the jump contour $\gamma$ such that it lies in  the shaded part $\mathcal H_+^s$ in the upper half plane, and we need the jump contour $-i\delta+\mathbb R$ to be contained in $\mathcal H_-^s$. 
We are able to deform the contours since the jump matrices $J_U(\zeta)$ are analytic for $\zeta\in\mathbb C\setminus\{\mp 2 i,\mp 2 i\mp\frac{2i}{\theta}, \mp 2i\mp\frac{4i}{\theta},\cdots\}$, but one should pay attention to the fact that we cannot cross singularities while deforming contours.

\medskip

Let us first consider the case $\theta=0$.
Then, for $s>1$, by Proposition \ref{prop:signh}, we can indeed take $\gamma$ in the shaded region $\mathcal H_+^s$, and we can deform $\mathbb R$ to a curve in the white region $\mathcal H_-^s$ in the lower half plane. 
For $s=1$, we can keep $\gamma$ in the shaded region $\mathcal H_+^1$ only if we let it go through $0$. Similarly, we then deform the horizontal line $-i\delta+\mathbb R$ to a contour in the white region $\mathcal H_-^1$ in the lower half plane but passing through $0$.
For $s<1$, we are not in such an ideal scenario, as we cannot immediately make a convenient choice for the jump contours. We will then be forced to let $\gamma$ coincide with the real line on an interval $[-b,b]$, before going into the upper half plane. Similarly, we will deform the horizontal line $-i\delta+\mathbb R$ to a suitable curve lying strictly in the lower half plane, except on the interval $[-b,b]$.

\medskip

For $\theta>0$ sufficiently small, the transition between the two phases will not happen at $s=1$, but rather at $s=-\theta\psi(\theta)$, which is bigger than $1$ for $\theta>0$ small and which converges to $1$ as $\theta\to 0$.
In each of the three cases $s>-\theta\psi(\theta)$, $s=-\theta\psi(\theta)$, $s<-\theta\psi(\theta)$, we define a jump contour $\Sigma$ as the closure of $\Sigma^+\cup\Sigma^0\cup\Sigma^-$, where $\Sigma^+$ is the part of $\Sigma$ lying in the upper half plane, $\Sigma^-$ the part of $\Sigma$ lying in the lower half plane, and $\Sigma^0=(-b,b)$ the (possibly empty) part of $\Sigma$ lying on the real line.
Depending on the value of $s$, we choose the jump contours as follows, and as illustrated in Figure \ref{figure:lenses2}.
\begin{itemize}
\item For $s>-\theta\psi(\theta)$: $\Sigma^0$ is empty, $\Sigma^+=\gamma$ is a closed curve in $\mathcal H_+^{s}$, and $\Sigma^-$ is a curve connecting $e^{i(\pi+\delta)}\infty$ with $e^{-i\delta}\infty$ through $\mathcal H_-^{s}$, with $\delta>0$ sufficiently small. Note that for any $\epsilon>0$, we can take $\Sigma^\pm$ independent of $s$ for all $s>-\theta\psi(\theta)+\epsilon$, since the sets $\mathcal H_\pm^s$ are increasing in $s$.
\item For $s=-\theta\psi(\theta)$: $\Sigma^0$ is empty, $\Sigma^+$ is a curve in $\mathcal H_+^1\cup\{0\}$ closing at the point $0$, and $\Sigma^-$ consists of a curve connecting $e^{i(\pi+\delta)}\infty$ with $0$ and one connecting $0$ with $e^{-i\delta}\infty$ through $\mathcal H_-^1\cup\{0\}$.
\item For $0<s<-\theta\psi(\theta)$: $\Sigma^0=(-b,b)$ with $b=b(s,\theta)>0$  to be determined later, $\Sigma^+$
is a curve connecting $b$ with $-b$ in the upper half plane, and $\Sigma^-$ consists of a curve connecting $e^{i(\pi+\delta)}\infty$ with $-b$ and one connecting $b$ with $e^{-i\delta}\infty$ in the lower half plane. 
\end{itemize}

\begin{figure}[H]	\begin{center}
		\begin{tikzpicture}			\node at (-6,0) {};
			\fill (-6,0) circle (0.05cm);
			\node at (-6+0.15,-0.2) {$0$};
			\fill (-6,1.5) circle (0.05cm);
			\node at (-6.15,1.3) {$2i$};
			\draw[dashed,->-=1,black] (-6,-1.5) to [out=90, in=-90] (-6,3.2);
			\draw[dashed,->-=1,black] (-8.4,0) to [out=0, in=-180] (-6+2.4,0);
						\node at (-6+1.6,-0.6) {\small $\Sigma^- $};			
			\draw[-<-=0.6,black] (-6,-0.4)--($(-6,-0.4)+(180+15:1.5)$);
		\draw[->-=0.6,black] (-6,-0.4)--($(-6,-0.4)+(-15:1.5)$);
			\node at (-7.5,-0.6) {\small$\Sigma^-$};
			\node at (-6+0.3,2.6) {\small$\Sigma^+$};
				\draw[->-=0.6,black] ([shift=(0.360:1.0cm)]-6.4,0.6) arc (-50.180:360:1cm);

			\node at (0,0) {};
			\fill (0,0) circle (0.05cm);
			\node at (0.15,-0.2) {$0$};
			\fill (0,1.5) circle (0.05cm);
			\node at (0.15,1.3) {$2i$};
			\draw[dashed,->-=1,black] (0,-1.5) to [out=90, in=-90] (0,3.2);
			\draw[dashed,->-=1,black] (-2.4,0) to [out=0, in=-180] (2.4,0);
			\draw[->-=0.6,black] (0,0)--(0.72,0.61);
			\draw[->-=0.6,black] (-0.70,0.61)--(0,0);
			\node at (1.5,-0.6) {\small$\Sigma^-$};			
					\draw[-<-=0.6,black] (0,0)--($(0,0)+(180+20:1.5)$);
		\draw[->-=0.6,black] (0,0)--($(0,0)+(-20:1.5)$);
			\node at (-1.5,-0.6) {\small$\Sigma^-$};
			\node at (0.3,2.7) {\small$\Sigma^+$};
				
				\draw[->-=0.6,black] ([shift=(0.180:1.0cm)]-0.3,0.6) arc (-50.180:230:1.08cm);

							\node at (6,0) {};
			\fill (6,0) circle (0.05cm);
			\node at (6.15,-0.2) {$0$};
			\fill (6,1.5) circle (0.05cm);
			\node at (6.15,1.3) {$2i$};
			\draw[dashed,->-=1,black] (6,-1.5) to [out=90, in=-90] (6,3.2);
			\draw[dashed,->-=1,black] (6-2.4,0) to [out=0, in=-180] (8.4,0);
			\draw[->-=0.6,black] (6,0)--(7,0);
			\fill (7,0) circle (0.05cm);
			\node at (6.9,0.2) {\small $b$};
			\node at (6-0.95,0.2) {\small $-b$};
			\node at (7.5,-0.6) {\small$\Sigma^-$};			
			\draw[-<-=0.5,black] (6,0)--(5,0);
			\fill (5,0) circle (0.05cm);
		\draw[-<-=0.6,black] (5,0)--($(5,0)+(180+30:1.5)$);
		\draw[->-=0.6,black] (7,0)--($(7,0)+(-30:1.5)$);
			\node at (6-1.5,-0.6) {\small$\Sigma^-$};
			\node at (6.3,2.6) {\small$\Sigma^+$};
				\node at (6-0.15,0.2) {\small$\Sigma^{0}$};
				\draw[->-=0.6,black] ([shift=(0.180:1.0cm)]6,0) arc (-50.180:230:1.58cm);$$
\end{tikzpicture}
		\caption{The shape of the jump contour $\Sigma=\Sigma^+\cup\Sigma^0\cup\Sigma^-$ for the RH problem for $T$, in the cases $s>-\theta\psi(\theta)$ (left), $-s= -\theta\psi(\theta)$ (middle), and $0<s<-\theta\psi(\theta)$ (right).\label{figure:lenses2}
}
	\end{center}
\end{figure}
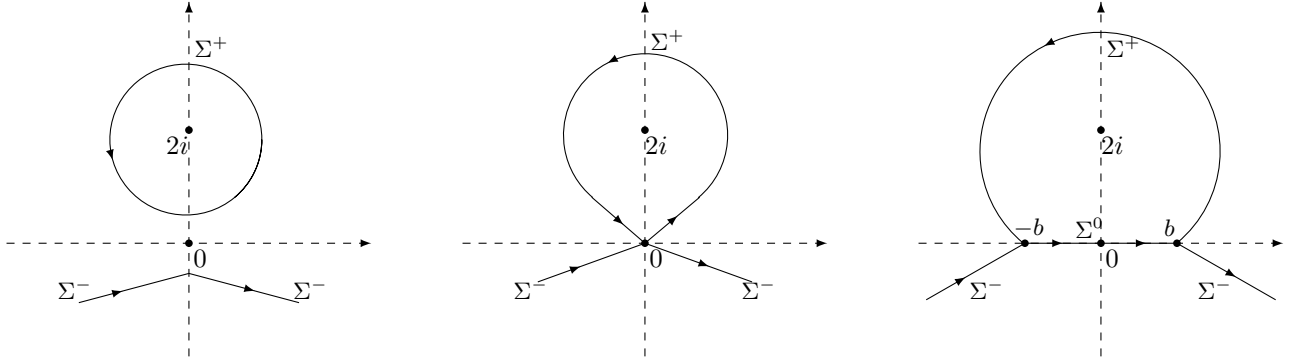

Let us assume, without loss of generality and for simplicity, that $\gamma$ is a small contour surrounding $2i$, such that it always lies in the interior of $\Sigma^+$.
In all three cases $s>-\theta\psi(\theta)$, $s= -\theta\psi(\theta)$, and $0<s<-\theta\psi(\theta)$, we define $T$ piecewise as follows: \begin{itemize}
\item in the region below $\Sigma^-$, it is the analytic continuation of $U$ from the region below $-i\delta+\mathbb R$;
\item in the region above $\Sigma^-$ and outside $\Sigma^+$, it is the analytic continuation of $U$ from the region above $-i\delta+\mathbb R$ and outside $\gamma$;
\item inside $\Sigma^-$, it is the analytic continuation of $U$ from the region inside $\gamma$.
\end{itemize}
In other words, we define
\be\label{def:Ttauneg}
T(\zeta)=\begin{cases}
U(\zeta)\begin{pmatrix}
				1 &  0 \\ e^{nh(\zeta;s,\theta)} & 1
			\end{pmatrix},&\mbox{for $\Im\zeta<-\delta$ and $\zeta$ above $\Sigma^-$,}\\
			U(\zeta)\begin{pmatrix}
				1 &  -e^{-nh(\zeta;s,\theta)}\\0 & 1
			\end{pmatrix},&\mbox{for $\zeta$ in between $\gamma$ and $\Sigma^+$,}\\
			U(\zeta)\begin{pmatrix}
				1 & 0\\ -e^{nh(\zeta;s,\theta)} & 1
			\end{pmatrix},&\mbox{for $\Im\zeta>-\delta$ and $\zeta$ below $\Sigma^-$,}\\
U(\zeta),&\mbox{elsewhere.}\end{cases}\ee

Then $T$ solves the following RH problem.
\subsubsection*{RH problem for $T$}
\begin{itemize}
	\item[(T1)] $T:\C \setminus \Sigma\to \C^{2 \times 2}$ is analytic.
	\item[(T2)] Across the contour $\Sigma$, $T$ satisfies the jump condition
	\begin{align*}
		T_+(\zeta) = T_-(\zeta)J_T(\zeta),
	\end{align*}
	where
	\begin{align*}
		J_T(\zeta) = \begin{cases}
			\begin{pmatrix}
				1 &  0 \\ e^{nh(\zeta;s,\theta)} & 1
			\end{pmatrix} &\text{ for $\zeta \in \Sigma^-$}, \\ 
			\begin{pmatrix}
				1 & -e^{-nh(\zeta;s,\theta)} \\ 0 & 1
			\end{pmatrix} &\text{ for $\zeta \in \Sigma^+ $}, \\
			\begin{pmatrix}
				1 & - e^{-n h(\zeta;s,\theta)} \\e^{n h(\zeta;s,\theta)}  & 0
			\end{pmatrix}  &\text{ for $\zeta \in \Sigma^0$}.
		\end{cases}
	\end{align*}
	\item[(T3)] $T(\zeta) = I -i \frac{Y_1(s,n,\theta)}{n\zeta} + \mathcal{O}(\zeta^{-2})$ as $\zeta \to \infty$.
\end{itemize}

In the case $s>-\theta\psi(\theta)+\epsilon$, as a consequence of Proposition \ref{prop:signh} and the construction of the jump contours, the jump matrices for $T$ converge to $I$ as $n\to\infty$ provided that $\theta>0$ is sufficiently small.

\begin{proposition}\label{prop:smallnormtauneg}
Let $\epsilon>0$. There exist $c,\theta_0>0$ such that 
\[J_T(\zeta;s,\theta)=I+O(e^{-c ns|\zeta|}),\qquad n\to\infty,\]
uniformly for $\zeta\in\Sigma$, $0\leq \theta<\theta_0$, $s>-\theta\psi(\theta)+\epsilon$.
\end{proposition}
\begin{proof}
Since $\Sigma^0$ is empty, it is sufficient to prove that there exists $c>0$ such that
\begin{align*}
&\Re h(\zeta;s,\theta)\geq cs|\zeta|,&\zeta\in\Sigma^+,\\
&\Re h(\zeta;s,\theta)\leq -cs|\zeta|,&\zeta\in\Sigma^-,
\end{align*}
for $0\leq \theta<\theta_0$, $s>-\theta\psi(\theta)+\epsilon$. 
Using the large $\zeta$ behavior of $h$, we can verify in a straightforward way that the second inequality holds for $|\zeta|>R$, $\zeta\in\Sigma^-$, for some sufficiently large $R>0$.
Next, we consider $|\zeta|\leq R$, and we can assume that this region contains the whole curve $\Sigma^+$. By construction, the above inequalities hold for $\theta=0$. Moreover, by continuity in $\theta\geq 0$, they continue to hold for $\theta$ sufficiently small.
\end{proof}
The above result implies that $\|J_T-I\|_1, \|J_T-I\|_2, \|J_T-I\|_\infty$ are $O(e^{-cns})$ as $n\to\infty$.
 It then follows from the general theory of RH problems, see \cite{Deift, DKMVZ1, DKMVZ2}, that the RH solution $T$ is also close to identity as $n\to\infty$.
\begin{corollary}\label{cor:tauneg}
Let $\epsilon>0$. There exist $c,\theta_0>0$ such that 
\[T(\zeta)=I+O\left(\frac{e^{-cns}}{|\zeta|+1}\right),\qquad n\to\infty,\]
uniformly for $\zeta\in\mathbb C\setminus\Sigma$, $0\leq \theta<\theta_0$, $s>-\theta\psi(\theta)+\epsilon$. In particular, $Y_1(s,n,\theta)=O(ne^{-cns})$ as $n\to\infty$, uniformly for $0\leq \theta<\theta_0$, $s>-\theta\psi(\theta)+\epsilon$.
\end{corollary}
This completes the asymptotic RH analysis in the case $s>-\theta\psi(\theta)+\epsilon$.

For $s$ close to $-\theta\psi(\theta)$, the convergence of the jump matrices to $I$ breaks down for $\zeta$ near $0$, and we will be required to construct a local parametrix in that case.

For
 $s<-\theta\psi(\theta)-\epsilon$, the jumps are not small in a bigger neighborhood of $0$, and we will need to construct a $g$-function to remedy this, before proceeding with the construction of a global parametrix and local parametrices.

\subsection{Construction of the $g$-function for $0<s<-\theta\psi(\theta)$}
Our next transformation of the RH problem will rely on the construction of a $g$-function. As usual in asymptotic RH analysis, the $g$-function is the fundamental object that enables one to arrive, after the construction of parametrices, at a small-norm RH problem.

The $g$-function $g(\zeta)=g(\zeta;s,\theta)$ will have to satisfy the following conditions, for some $b=b(s,\theta)\geq 0$ that we will fix later on.
\subsubsection*{RH conditions for $g$}\begin{itemize}
\item[(g1)] $g:\mathbb C\setminus[-b,b]\to\mathbb C$ is analytic,
\item[(g2)] $g_+(\zeta)+g_-(\zeta)=h(\zeta;s,\theta)$ for $\zeta\in(-b,b)$,
\item[(g3)] $g(\zeta)=\frac{g_1}{\zeta}+O(\zeta^{-2})$ as $\zeta\to\infty$, for some $g_1=g_1(s,\theta)$.
\end{itemize}

It is not hard to construct such a function $g$: we let 
\begin{align}\label{def:g}
	g(\zeta)= a(\zeta)\int_{-b}^{b} \frac{h(u;s,\theta)}{a_+(u)(u-\zeta)} \frac{du}{2 \pi i},
\end{align}
with
\begin{align*}
a(\zeta)=\left((\zeta-b)(\zeta+b)\right)^{1/2},
\end{align*}
where the branch is chosen such that $a$ is analytic in $\mathbb C\setminus[-b,b]$ and such that $a(\zeta) \sim \zeta$ as $\zeta \to \infty$, such that
$a_+(u)=i\sqrt{b^2-u^2}$ for $u\in(-b,b)$. Then, $g$ is analytic except on $[-b,b]$ and, by Cauchy's theorem, $g$ satisfies the jump relation (g2).
We also see that
\[\lim_{\zeta\to\infty}g(\zeta)=\frac{1}{2\pi}\int_{-b}^b\frac{h(u;s,\theta)}{\sqrt{b^2-u^2}}du=0,\]
since $h$ is antisymmetric.
Moreover, we obtain
\begin{align*}
g(\zeta;s,\theta) &= \frac{g_1(s,\theta)}{\zeta} + O(\zeta^{-2}) 
\end{align*}
as $\zeta \to \infty$, with
\begin{align} 
	g_1(s,\theta)&= \int_{-b}^{b} \frac{h(u;s,\theta) u}{\sqrt{b^2-u^2}} \frac{du}{2 \pi }\nonumber\\& = \int_{-b}^{b}h'(u;s,\theta) \sqrt{b^2-u^2} \frac{du}{2 \pi}\nonumber\\
	&=\frac{1}{\pi}\int_{0}^{b}h'(u;s,\theta) \sqrt{b^2-u^2} du\nonumber
\\&=\frac{b^2}{\pi}\int_{0}^{1}h'(bv;s,\theta) \sqrt{1-v^2} dv	
	=-i f(s,\theta),\label{eq: formula g1}
\end{align}
with $f$ as in \eqref{eq: def f}.

If we would proceed with arbitrary $b$, this $g$-function would not lead to jump matrices close to identity, and the local behavior of $g(\zeta)$ as $\zeta\to \pm b$ would not be suitable for the construction of local parametrices later on. We anticipate to this problem by choosing the value of $b=b(s,\theta)\geq 0$ appropriately. We define it, as announced in the introduction, to be the unique non-negative solution of the equation \eqref{eq: defining equation for b} if $0<s\leq -\theta\psi(\theta)$. We now prove the existence and uniqueness of $b(s,\theta)$.

\begin{proposition} \label{lemma: choice of b}
	For any $\theta>0$ and $s\in(0,-\theta\psi(\theta))$, there exists a unique number $b=b(s,\theta)>0$ such that \eqref{eq: defining equation for b} holds.
	Moreover, $b(s,\theta)$ is differentiable and strictly decreasing in $s$. As $s\to -\theta\psi(\theta)$, \eqref{eq:bastau} holds.
\end{proposition}
\begin{proof}
Recall the expressions \eqref{def:h}--\eqref{eq:der2h} for $h$ and its first two derivatives.
Observe that we can re-write \eqref{eq: defining equation for b} as
\begin{equation}\label{eq:endpointH}H(b):=\int_{0}^1 \frac{i h'(bu;s,\theta)}{\sqrt{1-u^2}}du=0.\end{equation} We have $H(0)=\frac{\pi}{2}(s+\theta\psi(\theta))<0$, and 
$H(\infty)=+\infty$, since the digamma function behaves like a logarithm at infinity. Moreover, $H$ is strictly increasing in $b>0$, since
\[H'(b)=\int_{0}^1 \frac{i uh''(bu)}{\sqrt{1-u^2}}du>0.\]
We conclude that there is a unique solution $b(s,\theta)$ for any $\theta>0$, $s\in(0,-\theta\psi(\theta))$. Since $ih'(bu;s,\theta)$ increases with $s$, $H(b)$ increases with $s$, such that $b(s,\theta)$ decreases with $s$. Moreover, by the implicit function theorem, $b(s,\theta)$ is differentiable in $s\in(0,\theta\psi(\theta))$. 
Taylor expanding $H$ as $b\to 0$, we obtain
\begin{align*}H(b)&=H(0)+H'(0)b+\frac{H''(0)}{2}b^2+O(b^3)\\
&=\frac{\pi}{2}\left(s+\theta\psi(\theta)\right)+\frac{b^2}{2}\int_{0}^1 \frac{i u^2h'''(0)}{\sqrt{1-u^2}}du+O(b^3)\\
&=\frac{\pi}{2}\left(s+\theta\psi(\theta)\right)-\frac{\pi\theta^3}{32}\psi''(\theta) b^2
+O(b^3).\end{align*}
It then follows that \eqref{eq:bastau} holds.
\end{proof}

\subsection{Transformation $T\mapsto S$}
We now define $S$ as follows:
\begin{align}\label{def:S}
	S(\zeta) =  T(\zeta) \begin{pmatrix}
		e^{-n g(\zeta)} & 0 \\0& e^{n g(\zeta)}  
	\end{pmatrix},
\end{align}
where $g(\zeta)=g(\zeta;s,\theta)$ is the $g$-function which we just constructed if $\epsilon<s<-\theta\psi(\theta)$, and $g(\zeta)=0$ if $-\theta\psi(\theta)\leq s\leq -\theta\psi(\theta)+\epsilon $.
To state the jump relations for $S$ in a simplified manner, it is convenient to introduce
the function $q(\zeta)=q(\zeta;s,\theta)$ as
\be\label{def:q}
q(\zeta)=g(\zeta;s,\theta)-\frac{1}{2}h(\zeta;s,\theta).\ee
This function is defined for $\zeta\in\mathbb C\setminus[-b,b]$ such that $h(\zeta)$ is defined. In particular, it is defined for $\zeta\in\Sigma^+\cup\Sigma^-$.
Using \eqref{def:g} and deforming the integration contour $[-b,b]$ to a positively oriented loop $\mathcal C$ around $[-b,b]$ and $\zeta$, we obtain the following integral representation for $q$:
\be
\label{eq:intq}
q(\zeta)=-a(\zeta)\int_{\mathcal C} \frac{h(u;s,\theta)}{a(u)(u-\zeta)} \frac{du}{4 \pi i}.
\ee
For $\theta=0$, we can compute $g'(\zeta;s,0)$ and $q'(\zeta;s,0)$ explicitly. We then have
\begin{align}&\label{def:g0}
g'(\zeta;s,0)=\frac{h'(\zeta;s,0)}{2}+\frac{is\zeta a(\zeta)}{2(\zeta^2+4)},\\
&\label{def:q0}q'(\zeta;s,0)=\frac{is\zeta a(\zeta)}{2(\zeta^2+4)}.
\end{align}

Using the definitions of $S$ and $q$, we obtain the following RH conditions for $S$.

\subsubsection*{RH problem for $S$}
\begin{itemize}
	\item[(S1)] $S:\C \setminus \Sigma \to \C^{2 \times 2}$ is analytic.
	\item[(S2)] Across the contour $\Sigma$, $S$ satisfies the jump condition
	\begin{align*}
		S_+(\zeta) = S_-(\zeta)J_S(\zeta),
	\end{align*}
	where
	\begin{align}\label{def:JS}
		J_S(\zeta) = \begin{cases}
			\begin{pmatrix}
				1 &  0 \\ e^{-2n q(\zeta)} & 1
			\end{pmatrix} &\text{ for $\zeta \in \Sigma^-$}, \\ 
			\begin{pmatrix}
				1 & -e^{2n q(\zeta)} \\ 0 & 1
			\end{pmatrix} &\text{ for $\zeta\in \Sigma^+ $}, \\
			\begin{pmatrix}
				e^{-n(q_+(\zeta)-q_-(\zeta))} & - 1\\ 1 & 0
			\end{pmatrix}  &\text{ for $\zeta\in \Sigma^0$}.
		\end{cases}
	\end{align}
	\item[(S3)] $S(\zeta) = I + \frac{-iY_1(s,n,\theta)-n^2g_1(s,\theta) \sigma_3}{n\zeta} + \mathcal{O}(\zeta^{-2})$ as $\zeta \to \infty$.
\end{itemize}
Recall that $\Sigma^0=(-b,b)$ for $\epsilon\leq s <-\theta\psi(\theta)$, while $\Sigma^0$ is empty for $-\theta\psi(\theta)\leq s \leq -\theta\psi(\theta)+\epsilon$.

In analogy with the regions $\mathcal H_\pm^s$, we define $\mathcal G_+^s$ for the subset of the upper half plane containing all $\zeta$ such that $\Re q(\zeta;s,0)<0$, and $\mathcal G_-^s=\overline{\mathcal G_+^s}$ for the subset of the lower half plane on which $\Re q(\zeta;s,0)>0$.
Let $\epsilon>0$. We now claim that for any $\epsilon\leq s<-\theta\psi(\theta)-\epsilon$, there exist $\theta_0>0$ and curves $\Sigma^\pm$ in $\mathcal G_\pm^s$ such that $\Sigma^+$ is a curve connecting $b$ with $-b$, and where $\Sigma^-$ consists of a curve connecting $e^{i(\pi+\delta)}\infty$ with $-b$ and a curve connecting $b$ with $e^{-i\delta}\infty$, for some small $\delta>0$. Moreover, we can take $\Sigma^\pm$ such that they make an angle $\pm\pi/3$ with the real line at $\pm b$. In other words, we can take $\Sigma^\pm$ as illustrated in Figure \ref{figure:lenses2} on the right.

To prove the claim, we proceed similarly as in the proof of Proposition \ref{prop:signh}. First, from the behavior of $g$ and $h$ as $\zeta\to\infty$, it follows that $\{\zeta\in\mathbb C: |\zeta|\geq R, \pm\Im \zeta>0\}$ is a subset of $\mathcal G_\pm^s$ for sufficiently large $R>0$. Moreover, from \eqref{def:q0}, it follows that $-iq'(\zeta;s,0)>0$ for $\zeta\in\mathbb R\setminus[-b,b]$. Using the Cauchy-Riemann equations, we can then again show that small sectors $0<\arg(\zeta-b)<\delta$ and $\pi-\delta<\arg(\zeta+b)<\pi$ belong to $\mathcal G_+^s$. Similarly, small sectors
 $-\delta<\arg(\zeta-b)<0$ and $\pi<\arg(\zeta+b)<\pi+\delta$ belong to $\mathcal G_-^s$.

By construction, we then obtain the following result.
\begin{proposition}\label{prop:smallnormtaupos}
Let $\epsilon >0$. There exist $\eta,c,\theta_0>0$ such that 
\begin{align*}&J_S(\zeta)=\begin{cases}
I+O(e^{-c n|\zeta|}),&\zeta\in\Sigma^\pm, \ |\zeta\pm b|>\eta, \\
\begin{pmatrix}0&-1\\1&0\end{pmatrix}+O(e^{-c n}),&\zeta\in\Sigma^0, \ |\zeta\pm b|>\eta,\end{cases}
\end{align*}
as $n\to\infty$, uniformly in $\zeta$ and in $0\leq \theta<\theta_0$, $\epsilon\leq s<-\theta\psi(\theta)-\epsilon$.
\end{proposition}
\begin{proof}
We first observe that
\begin{align*}
&(q_+-q_-)(\zeta;s,0)\geq c,&\zeta\in\Sigma^0,\ |\zeta\pm b|>\eta,\\
&\Re q(\zeta;s,0)\leq -c|\zeta|,&\zeta\in\Sigma^+,\ |\zeta \pm b|>\eta,\\
&\Re q(\zeta;s,0)\geq c|\zeta|,&\zeta\in\Sigma^-,\ |\zeta\pm b|>\eta,
\end{align*}
for $\epsilon\leq s<-\theta\psi(\theta)-\epsilon$, by construction of the contours, by \eqref{def:q0}, and by the fact that $q(\pm b;s,0)=0$. 
By continuity in $\theta\geq 0$, these inequalities still hold for $\theta$ sufficiently small.
\end{proof}

\section{Construction of parametrices for $\epsilon\leq s< -\theta\psi(\theta)-\epsilon$}
\label{section:RH2}
\subsection{The global parametrix}

We define $P^{\infty}$ as the solution of the formal large $n$ limit of the RH problem for $S$, obtained after ignoring exponentially small jumps and small neighborhoods of $\pm b$.

\subsubsection*{RH problem for $P^{\infty}$}
\begin{itemize}
	\item[($P^\infty$1)] $P^\infty:\C \setminus [-b,b] \to \C^{2 \times 2}$ is analytic.
	\item[($P^\infty$2)] On $(-b,b)$, $P^\infty$ satisfies the jump condition
	\begin{align*}
		P^\infty_+(\zeta) = P^\infty_-(\zeta)\begin{pmatrix}
0 & -1 \\ 1 & 0
\end{pmatrix}.
	\end{align*}
	\item[($P^\infty$3)] $P^\infty(\zeta) = I + \frac{P_1^\infty}{\zeta} + \mathcal{O}(\zeta^{-2})$ as $\zeta \to \infty$.
\end{itemize}

The solution to this RH problem is explicit and it is given by (see \cite[Chapter 7]{Deift} for a similar construction)
\begin{equation}\label{def:Pinf}
P^{\infty}(\zeta)=P^\infty(\zeta;b) = \frac{1}{2} \begin{pmatrix}\gamma(\zeta)+\gamma(\zeta)^{-1} & i \left(\gamma(\zeta)-\gamma(\zeta)^{-1}\right) \\-i \left(\gamma(\zeta)-\gamma(\zeta)^{-1}\right) & \gamma(\zeta)+\gamma(\zeta)^{-1}\end{pmatrix},
\end{equation}
where 
\begin{equation}
\gamma(\zeta)=\gamma(\zeta;b)=\left(\frac{\zeta-b}{\zeta + b}\right)^{1/4}
\end{equation}
with the fourth root defined and analytic on $\C \setminus [-b,b]$, with branch cut on $(-b,b)$, and such that $\gamma(\zeta)\to 1$ as $\zeta\to\infty$. 
Equivalently,
\begin{equation}\label{definition: Pinfty}
    P^{\infty}(\zeta)=\begin{pmatrix}1&-1\\-i&-i\end{pmatrix}
    \left(\frac{\zeta-b}{\zeta+b}\right)^{\sigma_3 /4}
    \begin{pmatrix}1&-1\\-i&-i\end{pmatrix}^{-1},\qquad\mbox{for $\zeta\in\mathbb C\setminus[-b,b]$.}
\end{equation}
We easily compute
\be
\label{eq:P1infty}
P_1^\infty=\begin{pmatrix}0&-\frac{i}{2}b\\\frac{i}{2}b&0\end{pmatrix}.
\ee

\subsection{Local Airy parametrices}

In small disks $U_{\pm b}$ of radius $\eta>0$ around the endpoints $-b$ and $b$, we need to construct local parametrices. These local parametrices need to satisfy the following conditions.
\subsubsection*{RH problem for $P^{(\pm b)}$}
\begin{itemize}
\item[(P1)] $P^{(\pm b)}:U_{\pm b}\setminus\Sigma\to\mathbb C^{2\times 2}$ is analytic.
\item[(P2)] $P^{(\pm b)}_+(\zeta)=P^{(\pm b)}_-(\zeta)J_S(\zeta)$ for $\zeta\in U_{\pm b}\cap\Sigma$.
\item[(P3)] On $\partial U_{\pm b}$, we have the (uniform in $\zeta$ and in $\epsilon\leq s\leq -\theta\psi(\theta)-\epsilon$) matching condition $P^{(\pm b)}(\zeta)=\left(I+O\left(\frac{1}{n}\right)\right)P^{(\infty)}(\zeta)$ as $n\to\infty$.
\end{itemize}

For the construction of parametrices, it is important to understand the $\zeta\to\pm b$ behavior of the function $q(\zeta;s,\theta)$ appearing in the jump matrices \eqref{def:JS}.	Using the fact that $q/a$ is an even function, we compute
\begin{align*}
\frac{q(\zeta)}{a(\zeta)} &= \int_{\mathcal C} \frac{h(u;s,\theta)}{a(u)(\zeta-u)} \frac{du}{4 \pi i} \\
 &= \int_{\mathcal C} \frac{h(u;s,\theta)u}{a(u)(\zeta^2-u^2)} \frac{du}{4 \pi i} \\
 &= \int_{\mathcal C} \frac{h(u;s,\theta)u}{a(u)\left((\zeta^2-b^2)-(u^2-b^2)\right)} \frac{du}{4 \pi i} \\
 &= -\int_{\mathcal C} \frac{h(u;s,\theta)u}{a(u)\left(u^2-b^2\right)}
 \left(1 + \frac{\zeta^2-b^2}{u^2-b^2} \right)  \frac{du}{4 \pi i}  + R_b(\zeta),
\end{align*}
where $R_b(\zeta)$ is an error term which is $O\left(a(\zeta)^4\right)$ as $\zeta\to \pm b$. The term
\begin{align*}
	 \int_{\mathcal C} \frac{h(u;s,\theta)u}{a(u)\left(u^2-b^2\right)}
	 \frac{du}{4 \pi i}
\end{align*}
vanishes: one easily sees this by integrating by parts \eqref{eq:endpointH}. Hence we have
\be\label{eq:qlocal}
q(\zeta)\sim ca(\zeta)^3,\qquad \zeta\to \pm b.\ee
Using this local behavior of $q$ like a $3/2$-power, we can construct local parametrices in terms of a model RH problem whose solution can be built out of the Airy function and its derivative. Such local Airy parametrices were constructed originally in \cite{BaikDeiftJohansson, Deift, DKMVZ1, DKMVZ2}, and their construction is by now standard.
We omit the details of this construction, and refer the interested reader to \cite[Section 3.5]{CGS17}, where Airy parametrices were constructed in a situation almost identical to ours. For our purposes, we will only need the existence of such local parametrices $P^{(\pm b)}$ satisfying conditions (P1)--(P3), the precise form of the parametrices is unimportant.

\subsection{The transformation $S \mapsto R$}

We define
\begin{align*}
R(\zeta) =  \begin{cases}
S(\zeta)P^{(\pm b)}(\zeta)^{-1},\quad &\text{ if $\zeta \in U_{\pm b}$}, \\
S(\zeta)P^\infty(\zeta)^{-1}, \quad &\text{ otherwise}.
\end{cases}
\end{align*}

\begin{figure}[H]\label{figure:SigmaR}
	\begin{center}
		\begin{tikzpicture}
							\node at (0,0) {};
		
			\node at (0.15,-0.2) {$0$};
			\fill (0,1.5) circle (0.05cm);
			\node at (0.15,1.3) {$2i$};
			\draw[dashed,->-=1,black] (0,-1.5) to [out=90, in=-90] (0,3.2);
			\draw[dashed,->-=1,black] (-2.4,0) to [out=0, in=-180] (2.4,0);
			\draw[->-=0.6,black] (-0.5,0)--(0.5,0);
			\fill (1,0) circle (0.05cm);
						\draw[-<-=0.5,black] (-1,0) circle (0.5cm);
									\draw[-<-=0.5,black] (1,0) circle (0.5cm);
			\node at (0.9,0.2) {\small $b$};
			\node at (-0.95,0.2) {\small $-b$};
			\fill (-1,0) circle (0.05cm);
		\draw[-<-=0.6,black] (-1.42,-0.25)--($(-1.42,-0.25)+(180+30:1.5)$);
		\draw[->-=0.6,black] (1.42,-0.25)--($(1.42,-0.25)+(-30:1.5)$);
				\draw[->-=0.6,black] ([shift=(10.100:1.47cm)]0,0) arc (-47.220:227:2.1cm);
\end{tikzpicture}
		\caption{Jump contour $\Sigma_R$ for $R$.}
	\end{center}
\end{figure}
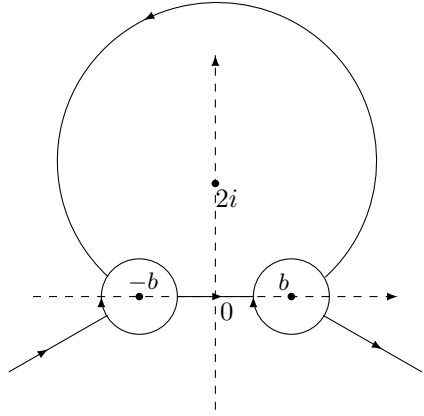

Then $R$ solves the following RH problem.

\subsubsection*{RH problem for $R$}
\begin{itemize}
	\item[(R1)] $R:\C \setminus \Sigma_R \to \C^{2 \times 2}$ is analytic, where $\Sigma_R=\left(\Sigma^+\cup\Sigma^-\cup\Sigma^0\cup\partial U_{b}\cup\partial U_{-b}\right)\setminus(U_b\cup U_{-b})$.
	\item[(R2)] Across the contour $\Sigma_R$, $R$ satisfies the jump condition
	\begin{align*}
		R_+(\zeta) = R_-(\zeta)J_R(\zeta),
	\end{align*}
	where
	\begin{align*}
		J_R(\zeta) = \begin{cases}
			P^\infty(\zeta)J_S(\zeta)P^\infty(\zeta)^{-1}  &\text{ for $\zeta \in \left(\Sigma^+\cup\Sigma^-\right)\setminus(U_b\cup U_{-b})$}, \\
			P^\infty(\zeta)J_S(\zeta)\begin{pmatrix}0&1\\-1&0\end{pmatrix}P^\infty(\zeta)^{-1}  &\text{ for $\zeta \in \Sigma^0\setminus(U_b\cup U_{-b})$}, \\
			P^{(\pm b)}(\zeta)P^\infty(\zeta)^{-1}  &\text{ for $\zeta \in \partial U_{\pm b} $}.
		\end{cases}
	\end{align*}
	\item[(R3)] $R(\zeta) = I + \frac{R_1(n,s,\theta)}{\zeta} + \mathcal{O}(\zeta^{-2})$ as $\zeta \to \infty$  with
	\begin{align}\label{eq:R1pos}
		R_1(s,n,\theta) = S_1(s,n,\theta)-P_1^\infty = \frac{-iY_1(s,n,\theta)-n^2g_1(s,\theta) \sigma_3}{n}-P_1^\infty.
	\end{align}
\end{itemize}

By construction, the jump matrix for $R$ is small as $n\to\infty$
\begin{proposition}\label{prop:smallnormtaupos2}
Let $\epsilon>0$. There exist $c,\theta_0>0$ such that 
\[J_R(\zeta)=I+O\left(\frac{1}{n(|\zeta|^2+1)}\right),\qquad n\to\infty,\]
uniformly for $\zeta\in\Sigma_R$, $0\leq \theta<\theta_0$, $\epsilon\leq s<-\theta\psi(\theta)-\epsilon$.
\end{proposition}
\begin{proof}
On $\partial U_{\pm b}$, the jump is $I+O(1/n)$ because of the matching condition (Pb3). Elsewhere, the jump matrix for $J_S$ is $I$ plus exponentially small terms which moreover decay as $\zeta\to\infty$, by Proposition \ref{prop:smallnormtaupos} and the boundedness of $P^\infty(\zeta)$.
\end{proof}

As in the case $s>-\theta\psi(\theta)+\epsilon$, it then follows from the general theory of RH problems that the RH solution $R$ is also close to identity as $n\to\infty$.
\begin{corollary}\label{cor:taupos}
Let $\epsilon>0$. There exist $\theta_0>0$ such that 
\[R(\zeta;s,n,\theta)=I+O\left(\frac{1}{n(|\zeta|+1)}\right),\qquad n\to\infty,\]
uniformly for $\zeta\in\mathbb C\setminus\Sigma_R$, $0\leq \theta<\theta_0$, $\epsilon\leq s<-\theta\psi(\theta)-\epsilon$. In particular, $R_1(s,n,\theta)=O(1/n)$ as $n\to\infty$,
uniformly for $0\leq \theta<\theta_0$, $\epsilon\leq s<-\theta\psi(\theta)-\epsilon$.
\end{corollary}
As a consequence, we have
\be\label{eq:asY1taupos0}
Y_1(s,n,\theta)=in^2g_1(s,\theta)\sigma_3+inP_1^\infty+O(1),\qquad n\to\infty\ee
uniformly for $0\leq \theta<\theta_0$, $\epsilon\leq s<-\theta\psi(\theta)-\epsilon$.
By \eqref{eq: formula g1} and \eqref{eq:P1infty}, we finally find
\be\label{eq:asY1taupos}
(Y_1(s,n,\theta))_{11}=n^2f(s,\theta)+O(1),\qquad n\to\infty\ee
uniformly for $0\leq \theta<\theta_0$, $\epsilon\leq s<-\theta\psi(\theta)-\epsilon$.

\section{Construction of parametrices for $-\theta\psi(\theta)-\epsilon\leq s\leq -\theta\psi(\theta)+\epsilon$}
\label{section:RH3}

Throughout this section, we assume that $\eta>0$ is sufficiently small, and we take $\epsilon>0$ sufficiently small such that $b(s,\theta)<\eta$ for $-\theta\psi(\theta)-\epsilon\leq s\leq -\theta\psi(\theta)+\epsilon$. We can do this by \eqref{eq:bastau}. We denote $U_0$ for the disk with radius $\eta$ centered at $0$.

\subsection{The global parametrix}
By Proposition \ref{prop:smallnormtau0}, we know that the jump matrix $J_S$ for $S$ is close to $I$ on $\Sigma^\pm\setminus U_0$ for $n$ large. For this reason, we need a global parametrix which is analytic outside $U_0$, and which tends to $I$ at infinity. While we could simply set $P^\infty(\zeta)=I$ at this point, this is not a convenient choice for $-\theta\psi(\theta)-\epsilon\leq s<-\theta \psi(\theta)$: indeed, this would not allow for a good matching with the local parametrix near $0$ later on.

For $-\theta\psi(\theta)-\epsilon\leq s<-\theta \psi(\theta)$, we rather define $P^\infty(\zeta)=P^\infty(\zeta;b)$ as before, given by \eqref{def:Pinf}, with $b=b(s,\theta)$ the unique solution of \eqref{eq: defining equation for b}. We recall that $P^\infty$ solves the RH conditions $(P^\infty 1)$--$(P^\infty 3)$, and that $b(s,\theta)>0$ is small for $s$ close to $-\theta\psi(\theta)$. For any $\eta>0$, we can choose $\epsilon>0$ sufficiently small such that $b(s,\theta)<\eta$, which implies that $P^\infty$ is analytic for $|\zeta|>\eta$.
 
For $-\theta\psi(\theta)\leq s\leq -\theta \psi(\theta)+\epsilon$, we set $P^\infty(\zeta)=I$. Note that this is the same as \eqref{eq: defining equation for b} with $b=0$.

\subsection{Model RH problem associated to the Painlev\'e II equation}

We will construct a local parametrix in $U_0$, which satisfies the same jump condition as $S$ and which matches with $P^\infty$ on the boundary of $U_0$.
\subsubsection*{RH problem for $P$}
\begin{itemize}
\item[(P1)] $P:\mathbb C\setminus U_{0}\to\mathbb C^{2\times 2}$ is analytic.
\item[(P2)] $P_+(\zeta)=P_-(\zeta)J_S(\zeta)$ for $\zeta\in U_{0}\cup\Sigma$.
\item[(P3)] On $\partial U_{0}$, we have the (uniform in $\zeta$, $0<\theta<\theta_0$, and in $-\theta\psi(\theta)-\epsilon\leq s\leq -\theta\psi(\theta)+\epsilon$) matching condition $P(\zeta)=\left(I+O\left({n^{-1/3}}\right)\right)P^{\infty}(z)$ as $n\to\infty$.
\end{itemize}

To constuct the local parametrix $P$, 
we will use the solution $\Psi(z)=\Psi(z;r)$ of a well-known model RH problem, connected to the Hastings-McLeod solution of the Painlev\'e II equation, and depending on a parameter $r\in\mathbb C$. This RH problem is due to Flaschka and Newell \cite{FlaschkaNewell} and was used and studied further in \cite{BaikDeiftJohansson, BleherIts, ClaeysKuijlaars, DZ2, FIKN}.

	\begin{figure}[H]
	\begin{center}
		\begin{tikzpicture}
			\node at (-6,0) {};
			\fill (-6,0) circle (0.05cm);
			\node at (-6+0.15,-0.2) {$0$};
	
			\draw[dashed,->-=1,black] (-6,-1.5) to [out=90, in=-90] (-6,1.5);
			\draw[dashed,->-=1,black] (-8.4,0) to [out=0, in=-180] (-6+2.4,0);
		\draw[-<-=0.6,black] (-6,0.4)--($(-6,0.4)+(180-15:1.5)$);
		\draw[->-=0.6,black] (-6,0.4)--($(-6,0.4)+(15:1.5)$);

			\draw[-<-=0.6,black] (-6,-0.4)--($(-6,-0.4)+(180+15:1.5)$);
		\draw[->-=0.6,black] (-6,-0.4)--($(-6,-0.4)+(-15:1.5)$);
			\node at (-7.5,-0.6) {\small$\widetilde\Sigma^-$};
			\node at (-7.5,0.9) {\small$\widetilde\Sigma^+$};

			\node at (0,0) {};
			\fill (0,0) circle (0.05cm);
			\node at (0.15,-0.2) {$0$};
			\draw[dashed,->-=1,black] (0,-1.5) to [out=90, in=-90] (0,1.5);
			\draw[dashed,->-=1,black] (-2.4,0) to [out=0, in=-180] (2.4,0);
	\draw[-<-=0.6,black] (0,0)--($(0,0)+(180-15:1.5)$);
		\draw[->-=0.6,black] (0,0)--($(0,0)+(15:1.5)$);			
			\draw[-<-=0.6,black] (0,0)--($(0,0)+(180+15:1.5)$);
		\draw[->-=0.6,black] (0,0)--($(0,0)+(-15:1.5)$);
			\node at (-1.5,-0.2) {\small$\widetilde\Sigma^-$};
			\node at (-1.5,0.5) {\small$\widetilde\Sigma^+$};

							\node at (6,0) {};
			\fill (6,0) circle (0.05cm);
			\node at (6.15,-0.2) {$0$};
			\draw[dashed,->-=1,black] (6,-1.5) to [out=90, in=-90] (6,1.5);
			\draw[dashed,->-=1,black] (6-2.4,0) to [out=0, in=-180] (8.4,0);
			\draw[->-=0.6,black] (6,0)--(7,0);
			\fill (7,0) circle (0.05cm);
			\node at (6.9,0.2) {\small $\beta$};
			\node at (6-0.95,0.2) {\small $-\beta$};
			\node at (4.5,-0.6) {\small$\widetilde\Sigma^-$};			
						\node at (4.5,0.7) {\small$\widetilde\Sigma^+$};				\node at (7.5,-0.6) {\small$\widetilde\Sigma^-$};			
						\node at (7.5,0.7) {\small$\widetilde\Sigma^+$};					\draw[-<-=0.5,black] (6,0)--(5,0);
			\fill (5,0) circle (0.05cm);
		\draw[-<-=0.6,black] (5,0)--($(5,0)+(180+30:1.5)$);
		\draw[->-=0.6,black] (7,0)--($(7,0)+(-30:1.5)$);
		\draw[-<-=0.6,black] (5,0)--($(5,0)+(180-30:1.5)$);
		\draw[->-=0.6,black] (7,0)--($(7,0)+(30:1.5)$);$$
\end{tikzpicture}
		\caption{The shape of the jump contour $\widetilde\Sigma$ for the RH problem for $\Psi$ (on the left), for the RH problem for $\widehat\Psi$ with $\beta=0$ (in the middle), and for $\widehat\Psi$ with $\beta>0$ (on the right). Compare this with Figure \ref{figure:lenses2}.}\label{fig:contourPsi2}	
	\end{center}
\end{figure}
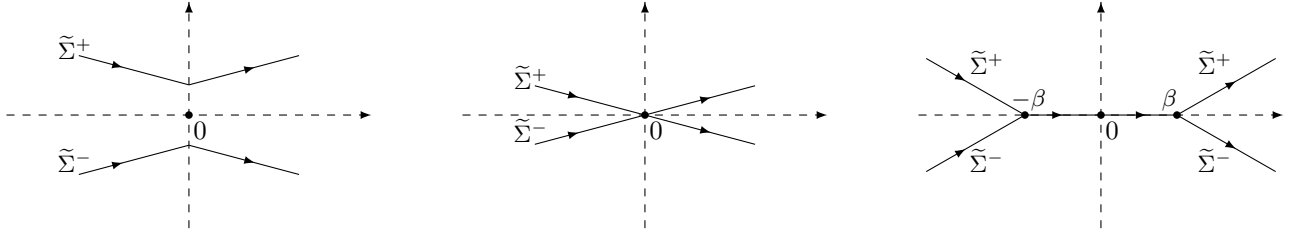

\subsubsection*{RH problem for $\Psi$}
\begin{itemize}
	\item[($\Psi$1)] $\Psi:\C \setminus \widetilde\Sigma$ is analytic, with $\widetilde\Sigma=\widetilde\Sigma^+\cup\widetilde\Sigma^-$, and $\widetilde\Sigma^+$ lies in the upper half plane, $\widetilde\Sigma^-$ lies in the lower half plane, and $\widetilde\Sigma$ lies for large $z$ within the region $\{\delta<|\arg z|<\frac{\pi}{3}-\delta\}\cup \{\delta<|\arg(-z)|<\frac{\pi}{3}-\delta\}$ for some $\delta>0$, like in the left picture of Figure \ref{fig:contourPsi2}.
	\item[($\Psi$2)] $\Psi$ satisfies the jump conditions
	\begin{align*}
	\Psi_+(z) = \Psi_-(z) = \begin{cases}
	\begin{pmatrix}
	1 & 0 \\ 1 & 1
	\end{pmatrix} \quad \text{ if $z \in \widetilde\Sigma^+$}, \\ 
	\begin{pmatrix}
	1 & -1 \\ 0 & 1
	\end{pmatrix} \quad \text{ if $z \in \widetilde\Sigma^-$}. 
	\end{cases} 
	\end{align*}
	\item[($\Psi$3)]As $z \to \infty$, we have
	\be\label{eq:asPsi}\Psi(z) e^{i(4/3z^3+rz)\sigma_3} = I +\Psi_1(r)z^{-1}+\Psi_2(r)z^{-2}+ O(z^{-3}),\ee for some matrices $\Psi_1(r),\Psi_2(r)$.
\end{itemize}

There is an equivalent RH problem, obtained by letting $\widetilde \Sigma^+$ and $\widetilde\Sigma^-$	
coincide on an interval $[-\beta,\beta]$ on the real line, with $\beta=\sqrt{-r/2}$, which will also be useful for us.	
	\subsubsection*{RH problem for $\widehat\Psi$}
\begin{itemize}
	\item[($\widehat\Psi$1)] $\widehat\Psi:\C \setminus \widetilde\Sigma$ is analytic, with $\widetilde\Sigma=\widetilde\Sigma^+\cup\widetilde\Sigma^-\cup[-\beta,\beta]$ as in the right micture of Figure \ref{fig:contourPsi2}.
	\item[($\widehat\Psi$2)] $\widehat\Psi$ satisfies the jump conditions
	\begin{align*}
	\widehat\Psi_+(z) = \widehat\Psi_-(z) = \begin{cases}
	\begin{pmatrix}
	1 & 0 \\ 1 & 1
	\end{pmatrix} \quad \text{ if $z \in \widetilde\Sigma^+$}, \\ 
	\begin{pmatrix}
	0 & -1 \\ 1 & 1
	\end{pmatrix} \quad \text{ if $z \in (-\beta,\beta)$}, \\ 
	\begin{pmatrix}
	1 & -1 \\ 0 & 1
	\end{pmatrix} \quad \text{ if $z \in \widetilde\Sigma^-$}. 
	\end{cases} 
	\end{align*}
	\item[($\widehat\Psi$3)]As $z \to \infty$, we have 
	\be\label{eq:ashatPsi}\widehat\Psi(z) e^{i(4/3z^3+rz)\sigma_3} = I +\Psi_1(r)z^{-1}+\Psi_2(r)z^{-2}+ O(z^{-3}).\ee
	\item[($\widehat\Psi$4)] $\widehat\Psi$ is bounded near $\pm\beta$.
\end{itemize}

These two equivalent RH problems have been studied intensively in the past decades and are by now well understood. Below, we summarize the relevant known results about its solution, for which we refer the reader to \cite{FIKN, FlaschkaNewell, ClaeysKuijlaars}.

First, we have that the solutions $\Psi,\widehat\Psi$ exist for any $s$ in a complex neighborhood of the real line, and that the matrices $\Psi_1(r),\Psi_2(r)$ are independent of the choice of jump contour $\widetilde\Sigma$. 
Secondly, the RH solution is connected to the Painlev\'e II equation in the following sense. We have
\be\label{def:Psi1}\Psi_1(r)=\frac{1}{2i}\begin{pmatrix}u(r)&q(r)\\-q(r)&-u(r)\end{pmatrix},\ee
	where $q$ is the Hastings-McLeod solution of the Painlev\'e II equation, i.e. 
	\[q''(r)=rq(r)+2q(r)^3\quad \mbox{ and }q(r)\sim \Ai(r)\ \mbox{ as }r\to +\infty,\] and $u(r)=q'(r)^2-sq(r)^2-q(r)^4$, such that $u'(r)=-q(r)^2$. 
	
	The expansion \eqref{eq:asPsi} in ($\Psi$3) is uniform for $r$ in compact subsets of the complex $r$-plane in which the RH problem is solvable, and it is also uniform as $r \to \infty$ with $|\arg r|<\delta$ with $\delta>0$ small. As $r\to-\infty$ or $r\to\infty$ with $|\arg r-\pi|<\delta$, it is not uniform, but one can rewrite it in order to obtain a uniform expansion.
As	
 $r\to\infty$ with $|\arg r-\pi|<\delta$, we have the following asymptotics, uniform for $|z|>\sqrt{-r}$ (see for instance \cite[Chapter 8]{FIKN} or \cite{DZ2}):
\be\label{eq:Psi32}\widehat\Psi(z,r) e^{\frac{4i}{3} (z^2+r/2)^{3/2}\sigma_3} =\left(I +O\left(\frac{1}{(|r|+1)z}\right)\right)P^\infty(z;\sqrt{-r/2}),\qquad z\to\infty,\ee 
with $P^\infty(z;\sqrt{-r/2})$ given by \eqref{def:Pinf} with $b=\sqrt{-r/2}$.

\subsection{Construction of the local parametrix for $-\theta\psi(\theta)\leq s \leq -\theta\psi(\theta)+\epsilon$}

We first record the following result stating that the jump matrix  $J_T(\zeta)$ is close to $I$ as $n\to\infty$, except for $\zeta$ close to $0$. We omit the proof, since it is the same as that of Proposition \ref{prop:smallnormtauneg}, for $|\zeta|>\rho$.
\begin{proposition}\label{prop:smallnormtau0T}
Let $\rho >0$. There exist $\epsilon,c,\theta_0>0$ such that 
\[J_T(\zeta)=
I+O(e^{-c n|\zeta|}),\qquad\zeta\in\Sigma^\pm, \ |\zeta|>\rho, 
\]
as $n\to\infty$, uniformly in $\zeta$ and in $0\leq \theta<\theta_0$, $-\theta\psi(\theta)\leq s<-\theta\psi(\theta)+\epsilon$.
\end{proposition}

This is not sufficient to conclude that $T$ is close to $I$, because the convergence of the jump matrices breaks down near $0$. Therefore, we need to construct a local parametrix near $0$, for which we will use $\Psi(z;r)$.

We will construct $P$ of the following form:
\begin{align}\label{def:PPII}
P(\zeta) = \begin{pmatrix}0&1\\-1&0\end{pmatrix}\Psi\big( n^{1/3} f(\zeta),r_n(\zeta) \big) 
\begin{pmatrix}0&-1\\1&0\end{pmatrix}
 e^{\frac{n}{2}h(\zeta)\sigma_3},
\end{align}
where $q$ is defined by \eqref{def:q}.
Here $f$ will be a conformal map in $U_0$ with $f(0)=0$ and $f'(0)>0$, and $r_n$ will be analytic in $U_0$ and uniformly bounded in $n,s,\zeta$.
First, we need $P$ to have the same jump contour as $S$, and therefore we need to take $\Sigma,\widetilde\Sigma$ and $f$ such that $n^{1/3}f\left(\Sigma\cap U_0\right)\subset \widetilde\Sigma$. If this holds, then it is straightforward to check that $P$ and $\Psi$ satisfy the same jump relations on $\Sigma\cap U_0$.
The remaining freedom to define $f$ and $r_n$ will be needed to achieve the matching condition (P3).
By condition ($\Psi$3) and \eqref{def:PPII}, it follows that 
the matching condition 
\be\label{eq:matching1}
P(\zeta)P^\infty(\zeta)^{-1}=P(\zeta)=I+O(n^{-1/3}),\qquad n\to\infty,
\ee
holds uniformly for $\zeta\in\partial U_0$, $-\theta\psi(\theta)\leq s\leq -\theta\psi(\theta)+\epsilon$, provided that we have the identity
\be \label{eq:identitylocal}n h(\zeta;s,\theta) = -\frac{8i}{3}n f(\zeta)^3-2in^{1/3}r_n(\zeta)f(\zeta).
\ee
To achieve this, we define
\be\label{def:fr}
f(\zeta)=\left(\frac{3i}{8}h(\zeta;-\theta\psi(\theta),\theta)\right)^{1/3},\qquad r_n(\zeta)=n^{2/3}(s+\theta\psi(\theta))\frac{\zeta}{2f(\zeta)}.
\ee It is then straightforward to check that \eqref{eq:identitylocal} holds, by \eqref{def:h}.
Since 
\[h(0;-\theta\psi(\theta),\theta)=h'(0;-\theta\psi(\theta),\theta)=h''(0;-\theta\psi(\theta),\theta)=0,\quad h'''(0;-\theta\psi(\theta),\theta)=\frac{i\theta^3\psi''(\theta)}{4},\]
we also have, as required, that $f$ is a conformal map, with  \be\label{eq:f0}
f(0)=0,\qquad f'(0)=\frac{\theta(-\psi''(\theta))^{1/3}}{4}>0.
\ee
Furthermore, $r_n$ is analytic near $0$, and we have 
\be\label{eq:rn0}r_n(0)=\frac{2}{\theta(-\psi''(\theta))^{1/3}}(s+\theta\psi(\theta))n^{2/3}\geq 0.\ee

\subsection{Construction of the local parametrix for $-\theta\psi(\theta)-\epsilon\leq s < -\theta\psi(\theta)$}

For $s+\theta\psi(\theta)=O(n^{-2/3})$ as $n\to\infty$, we could have proceeded directly with the construction of a local Painlev\'e II parametrix like before, without the need to use a $g$-function transformation, following the method of \cite{ClaeysKuijlaars}. However, since we need uniform asymptotics for $-\theta\psi(\theta)-\epsilon\leq s\leq -\theta\psi(\theta)+\epsilon$, we are forced to use a $g$-function transformation here as well, similarly to \cite{BaikDeiftJohansson}.
As already mentioned, we let $\eta>0$ be a sufficiently small number, and we let $\epsilon>0$ be such that $|b(s,\theta)|<\eta$ for $-\theta\psi(\theta)-\epsilon\leq s< -\theta\psi(\theta)$. 
We first prove that the jump matrix $J_S(\zeta)$ is close to $I$ as $n\to\infty$, except for $\zeta$ close to $0$.
\begin{proposition}\label{prop:smallnormtau0}
Let $\eta >0$. There exist $\epsilon,c,\theta_0>0$ such that 
\[J_S(\zeta)=
I+O(e^{-c n|\zeta|}),\qquad\zeta\in\Sigma^\pm, \ |\zeta|>\eta, 
\]
as $n\to\infty$, uniformly in $\zeta$ and in $0\leq \theta<\theta_0$, $-\theta\psi(\theta)-\epsilon\leq s<-\theta\psi(\theta)+\epsilon$.
\end{proposition}
\begin{proof}
For $\theta=0$ and $s=1$, we have
\[q(\zeta;1,0)=-\frac{1}{2}h(\zeta;1,0)=-\frac{1}{2}\log\frac{\zeta-2i}{\zeta+2i}+\frac{i}{2}\zeta.\]
Recall that we have chosen the contours $\Sigma^\pm$ such that 
\[\Re q(\zeta;1,0)\geq C,\qquad \zeta\in\Sigma^+,|\zeta|>\eta,\]
and
\[\Re q(\zeta;1,0)\leq -C,\qquad \zeta\in\Sigma^-,|\zeta|>\eta,\]
for some $C>0$.
From the behavior of $q$ at infinity, it follows moreover that we can replace $C$ by $c|\zeta|$ with $c>0$ sufficiently small in the above estimates.
By continuity of $q(\zeta;s,\theta)$ in $s,\theta$ for $|\zeta|>\eta$, it follows that these estimates continue to hold for $0<\theta<\theta_0$ with $\theta_0>0$ sufficiently small, and $s$ sufficiently close to $1$. Note moreover that any interval containing $1$ contains also $-\theta\psi(\theta)$ for $\theta>0$ small, since $\lim_{\theta\to 0} (-\theta\psi(\theta))=1$. The result follows.
\end{proof}

We already know that $q$ behaves like a $3/2$-power near $\pm b$, but we are now in a situation where $b(s,\theta)$ can converge to $0$ in the limit where $n\to\infty$. This is why we now need uniform control of $q$ in a fixed neighborhood of $0$, containing both $\pm b$.

\begin{proposition} \label{cor: p positive for small b and delta}
There exist $\eta,\epsilon,\theta_0>0$ such that 	\begin{align*}
	\frac{-i q(\zeta)}{a(\zeta)^3} &= c(s)+ O(a(\zeta)^2),
	\end{align*}
	uniformly for $0<\theta<\theta_0$, $-\theta\psi(\theta)-\epsilon\leq s<-\theta\psi(\theta)$, and $|\zeta|\leq \eta$. Furthermore,
	\begin{align*}
		c(s) =  - \theta^3\psi''(\theta)/48 +O(s+\theta\psi(\theta)), \quad s\to -\theta\psi(\theta).
	\end{align*}
\end{proposition}
\begin{proof}
We already computed (see the computation before \eqref{eq:qlocal}) 
\begin{align*}
\frac{-iq(\zeta)}{a(\zeta)} = a(\zeta)^2\int_{\mathcal C} \frac{h(u)u}{a(u)(u^2-b^2)^2}
   \frac{du}{4 \pi }  + O\left(a(\zeta)^4\right),
\end{align*}
uniformly for $|\zeta|<\eta$ with $\eta$ sufficiently small. We define
\begin{align*}
c(s)=\int_{\mathcal C} \frac{h(u)u}{a(u)\left(u^2-b^2\right)^2} \frac{du}{4 \pi }.
\end{align*}
Then, as $b=b(s,\theta)\to 0$, we find by a residue computation that
\[c(s)=\int_{\mathcal C} \frac{h(u)}{u^4} \frac{du}{4 \pi }+O(b^2)=-\frac{\theta^3\psi''(\theta)}{48}+O(b^2).\]
By \eqref{eq:bastau}, the error term is $O(s+\theta\psi(\theta))$ as $n\to\infty$, and the result follows. 
\end{proof}

We will now construct $P$ of the following form, similar to \eqref{def:PPII} but with $\widehat\Psi$ instead of $\Psi$ and with $q$ instead of $h$:
\begin{align}\label{def:PPII2}
P(\zeta) = \begin{pmatrix}0&1\\-1&0\end{pmatrix}\widehat\Psi\big( n^{1/3} f(\zeta),r_n(\zeta) \big) 
\begin{pmatrix}0&-1\\1&0\end{pmatrix}
 e^{-n q(\zeta) \sigma_3}.
\end{align}
Here $f$ will be the same as in \eqref{def:fr}, and $r_n(\zeta)$ will be analytic in $U_0$ and uniformly bounded.
As before, we will take $\Sigma,\widetilde\Sigma$ and $f$ such that $n^{1/3}f\left(\Sigma\cap U_0\right)\subset \widetilde\Sigma=\Sigma^+\cup\Sigma^-\cup[-\beta,\beta]$. 
This fixes in particular the value of $\beta=n^{1/3}f(b(s,\theta))$.

If is then again straightforward to verify that $P$ and $\Psi$ satisfy the same jump relations on $\Sigma\cap U_0$.
The remaining freedom to define $r_n$ will be needed to achieve the matching condition (P3).
By \eqref{eq:Psi32} and \eqref{def:PPII2}, it follows that 
the matching condition \be\label{eq:matching2}
P(\zeta)P^\infty(\zeta)^{-1}=I+O\left(\frac{1}{n^{1/3}(|r_n(\zeta)|+1)}\right),\qquad n\to\infty,
\ee holds uniformly for $\zeta\in\partial U_0$, $-\theta\psi(\theta)-\epsilon< s<-\theta\psi(\theta)$, is valid provided that we have the identity
\be \label{eq:identitylocal2}n q(\zeta;s,\theta) = \frac{4i}{3} \left(n^{2/3}f(\zeta)^2+\frac{1}{2}r_n(\zeta)\right)^{3/2}.
\ee
To achieve this, we define
\be\label{def:r2}
r_n(\zeta)=2\left(\frac{-3i}{4}nq(\zeta;s,\theta)\right)^{2/3}-2n^{2/3}f(\zeta)^2,
\ee 
which we can rewrite as
\[r_n(\zeta)=2n^{2/3}\left(\frac{-3i}{4}q(\zeta;s,\theta)\right)^{2/3}-2n^{2/3}\left(\frac{-3i}{4}q(\zeta;-\theta\psi(\theta),\theta)\right)^{2/3}\]
By Proposition \ref{cor: p positive for small b and delta}, $r_n$ is analytic in $U_0$, and using also \eqref{eq:bastau}, we obtain
\begin{align*}r_n(\zeta)&=\frac{b^2}{8}\theta^2(-\psi''(\theta))^{2/3}n^{2/3}+O(n^{2/3}b^4)+O(n^{2/3}b^2\zeta^2)\\
&=\frac{2}{\theta(-\psi''(\theta))^{1/3}}n^{2/3}(s+\theta\psi(\theta))\left(1+O(s+\theta\psi(\theta))+O(\zeta^2)\right),\end{align*}
uniformly in $\zeta\in U_0$, $0<\theta<\theta_0$, $-\theta\psi(\theta)-\epsilon\leq s<-\theta\psi(\theta)$.
For $\epsilon>0$ small, it follows that $r_n(\zeta)$ lies in a small sector around the negative real line, hence the RH problem for $\widehat\Psi\left(.;r=r_n(\zeta)\right)$ is solvable, and we can indeed use the asymptotic relation \eqref{eq:Psi32}.

\subsection{The transformation $S \mapsto R$}
We define
\begin{align*}
R(\zeta) = T(\zeta) \begin{cases}
P(\zeta)^{-1},\quad &\text{ if $\zeta \in U_0$}, \\
I, \quad &\text{ otherwise},
\end{cases}\ \mbox{if $-\theta\psi(\theta)\leq s \leq -\theta\psi(\theta)+\epsilon$,}
\end{align*}
and
\begin{align*}
R(\zeta) = S(\zeta) \begin{cases}
P(\zeta)^{-1},\quad &\text{ if $\zeta \in U_0$}, \\
P^\infty(\zeta;b(s,\theta))^{-1}, \quad &\text{ otherwise},
\end{cases}\ \mbox{if $-\theta\psi(\theta)-\epsilon\leq s < -\theta\psi(\theta)$.}
\end{align*}

\begin{figure}[H]\label{figure:SigmaR2}
	\begin{center}
		\begin{tikzpicture}
							\node at (0,0) {};
			\fill (0,0) circle (0.05cm);
			\node at (0.15,-0.2) {$0$};
			\fill (0,1.5) circle (0.05cm);
			\node at (0.15,1.3) {$2i$};
			\draw[dashed,->-=1,black] (0,-1.5) to [out=90, in=-90] (0,3.2);
			\draw[dashed,->-=1,black] (-2.4,0) to [out=0, in=-180] (2.4,0);

					\draw[-<-=0.5,black] (0,0) circle (0.9cm);			
		
			\draw[->-=0.6,black] (-0.8,-0.37)--($(-0.8,-0.37)+(180+30:1.5)$);
		\draw[->-=0.6,black] (0.8,-0.37)--($(0.8,-0.37)+(-30:1.5)$);
				\draw[->-=0.6,black] ([shift=(10.100:0.925cm)]0,0) arc (-56.220:236:1.6cm);
\end{tikzpicture}
		\caption{Jump contour $\Sigma_R$ for $R$.}
	\end{center}
\end{figure}
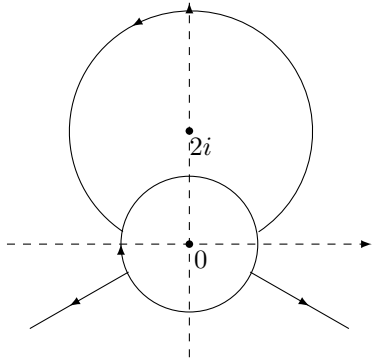

By construction, $R$ solves the following RH problem.
\subsubsection*{RH problem for $R$}
\begin{itemize}
	\item[(R1)] $R:\C \setminus \Sigma_R \to \C^{2 \times 2}$ is analytic.
	\item[(R2)] Across the contour $\Sigma_R$, $R$ satisfies the jump condition
	\begin{align*}
		R_+(\zeta) = R_-(\zeta)J_R(\zeta),
	\end{align*}
	where
	\begin{align}\label{eq:jumpR}
		J_R(\zeta) = I+O(n^{-1/3})
	\end{align}
	for $\zeta$ on $\Sigma_R$.
	\item[(R3)] $R(\zeta) = I + \frac{R_1(s,n,\theta)}{\zeta} + \mathcal{O}(\zeta^{-2})$ as $\zeta \to \infty$  with
	\begin{align*}
		R_1(s,n,\theta) = S_1(s,n,\theta) -P_1^\infty=\frac{-i Y_1(s,n,\theta)-n^2g_1(s,\theta)1_{(0,-\theta\psi(\theta))}(s) \sigma_3}{n}-P_1^\infty 1_{(0,-\theta\psi(\theta))}(s).
	\end{align*}
\end{itemize}

As before, it then follows that
\begin{align*}
	R(\zeta) = I + O\left(\frac{1}{n^{1/3}(|z|+1)}\right)
\end{align*}
as $n\to +\infty$, uniformly in $z\in\mathbb C\setminus\Sigma_R$.

It follows that 
\be\label{eq:asR1PII}
R_1(s,n,\theta)=O(n^{-1/3}),\qquad n\to\infty.
\ee
Hence,
\be
Y_1(s,n,\theta)=in^2g_1(s,\theta)1_{(0,-\theta\psi(\theta))}(s)\sigma_3+inP_1^\infty 1_{(0,-\theta\psi(\theta))}(s)+O(n^{2/3}),
\ee
and
\be\label{eq:asY1PII}
\left(Y_1(s,n,\theta)\right)_{11}=in^2g_1(s,\theta)1_{(0,-\theta\psi(\theta))}(s)+O(n^{2/3})=n^2f(s,\theta)1_{(0,-\theta\psi(\theta))}(s)+O(n^{2/3}),
\ee
as $n\to\infty$, uniformly for $-\theta\psi(\theta)-\epsilon\leq s\leq -\theta\psi(\theta)+\epsilon$.

\section{Proof of Theorem \ref{theorem:main}}\label{section:RHproof}
We recall the result from Proposition \ref{prop:diffid}:
\[\frac{d}{ds} \log Q_n^\theta(s) = \left(Y_1(s,n,\theta)\right)_{11},\]
and we substitute the results from the large $n$ asymptotic RH analysis in the right hand side. Let $\epsilon>0$ be sufficiently small.
For $s>-\theta\psi(\theta)+\epsilon$, we have by Corollary \ref{cor:tauneg} that $Y_1(s,n,\theta)=O(ne^{-cns})$ as $n\to\infty$;
for $-\theta\psi(\theta)-\epsilon\leq s\leq -\theta\psi(\theta)+\epsilon$, we have \eqref{eq:asY1PII}; finally, for $\epsilon<s<-\theta\psi(\theta)-\epsilon$, we have \eqref{eq:asY1taupos}.
We thus obtain the uniform in $s$ large $n$ asymptotics
\[\frac{d}{ds} \log Q_n^\theta(s) = \begin{cases}
O(ne^{-cns}),&s>-\theta\psi(\theta)+\epsilon,\\
n^2f(s,\theta)1_{(0,-\theta\psi(\theta))}(s)+O(n^{2/3}),&-\theta\psi(\theta)-\epsilon\leq s\leq -\theta\psi(\theta)+\epsilon,\\
n^2f(s,\theta)+O(1),&\epsilon<s\leq -\theta\psi(\theta)-\epsilon,
\end{cases}\ \mbox{as $n\to\infty$.}\]
Integrating in $s$ from $s$ to $+\infty$ and observing that $\log Q_n^\theta(+\infty)=0$, we obtain
\[\log Q_n^\theta(s) = \begin{cases}
O(e^{-c'ns}),&s>-\theta\psi(\theta)+\epsilon,\\
-n^2F(s,\theta)1_{(0,-\theta\psi(\theta))}(s)+O(n^{2/3}),&-\theta\psi(\theta)-\epsilon\leq s\leq -\theta\psi(\theta)+\epsilon,\\
-n^2F(s,\theta)+O(n^{2/3}),&\epsilon<s\leq -\theta\psi(\theta)-\epsilon,
\end{cases}\]
as $n\to\infty$.
This proves Theorem \ref{theorem:main}. The RH analysis allows in fact to obtain much stronger large $n$ asymptotics for $\log Q_n^\theta(s)$. For instance, we can expand the jump matrix for $R$ as far as we want in powers of $n$ as $n\to\infty$, and use this expansion to compute subleading asymptotic terms for $R$, and thus for $Y$. In this way, after long but straightforward computations, we could prove the stronger asymptotics
\be\label{eq:asQstrong}
\log Q_n^\theta(s)=
-\log F_{\rm TW}\left(\frac{2n^{2/3}(s+\theta\psi(\theta))}{\theta(-\psi''(\theta))^{1/3}}\right)+o(1),
\ee
as $n\to\infty$ and $s\to -\theta\psi(\theta)$ with $n^{2/3}(s+\theta\psi(\theta))=O(1)$.
These stronger asymptotics can be used to reprove \eqref{eq: central limit theorem}, but would not lead to a stronger result about the probability distribution of the log-Gamma polymer in Conjecture \ref{conjecture:main}.

\section{Proof of Conjecture \ref{conjecture:main} under Ansatz \ref{ansatz1}}\label{section:proofconjecture}
Assuming that Ansatz \ref{ansatz1} holds, we now know from \eqref{eq:detid3}, \eqref{def:Fredholmdetsmooth}, and Theorem \ref{theorem:main} that \be\lim_{n\to\infty}\frac{1}{n^2}\log\mathbb E\left(e^{-e^{-\frac{2n}{\theta}s}Z_n(\theta)}\right)=\lim_{n\to\infty}\frac{1}{n^2}\log\widetilde Q_n^\theta(s)=\lim_{n\to\infty}\frac{1}{n^2}\log Q_n^\theta(s)=-F(s,\theta),\label{eq:LDPQntilde}\ee
where $Z_n(\theta) $ is the log-Gamma polymer partition function associated with $\theta$ and $n$. From this, we can draw conclusions for the log-Gamma polymer partition function. We are interested in the large $n$ behavior of the probability
\begin{align*}
	\mathbb{P}\left[\log Z_n(\theta) \le \frac{2n}{\theta}s\right].
\end{align*}
For any $s>s_0\in\mathbb R$, we have the inequalities
\begin{align*}
	&\mathbb{E} e^{-e^{\log Z_n(\theta)-\frac{2n}{\theta}s_0}}\\ & = 	\mathbb{E}\left[ 1_{\{\log Z_n(\theta)\leq \frac{2n}{\theta}s\}}e^{-e^{\log Z_n(\theta)-\frac{2n}{\theta}s_0}} +1_{\{\log Z_n(\theta)>  \frac{2n}{\theta}s\}}e^{-e^{\log Z_n(\theta)-\frac{2n}{\theta}s_0}} \right] \\
	& \le\mathbb{P} \left[\log Z_n(\theta)\le  \frac{2n}{\theta}s\right] +\mathbb{E}\left[ 1_{\{\log Z_n(\theta)>  \frac{2n}{\theta}s\}}e^{-e^{\log Z_n(\theta)-\frac{2n}{\theta}s_0}} \right].
\end{align*}
Note that, whenever $\log Z_n(\theta)> \frac{2n}{\theta}s$, it holds that
\begin{align*}
	{e^{\log Z_n(\theta)-\frac{2n}{\theta}s_0}}  = {e^{\log Z_n(\theta)-\frac{2n}{\theta}s}}   e^{\frac{2n}{\theta}(s-s_0)} \ge e^{\frac{2n}{\theta}(s-s_0)} 
\end{align*}
and thus
\begin{align*}
e^{-e^{\log Z_n(\theta)-\frac{2n}{\theta}s_0}} \le e^{-e^{\frac{2n}{\theta}(s-s_0)}},
\end{align*}
which implies
\[\mathbb{P} \left[\log Z_n(\theta)\le  \frac{2n}{\theta}s\right] \ge \mathbb{E} e^{-e^{\log Z_n(\theta)-\frac{2n}{\theta}s_0}} - e^{-e^{\frac{2n}{\theta}(s-s_0)}}.\]
We can now set, for instance, $s=s_0+\frac{3\theta\log n}{2n}$, such that
\[e^{-e^{\frac{2n}{\theta}(s-s_0)}}=e^{-e^{3\log n}}=e^{-n^3},\]
and for sufficiently large $n$,
\begin{align}
	\mathbb{P} \left[\log Z_n(\theta)\le  \frac{2n}{\theta}s\right] \ge \mathbb{E} e^{-e^{\log Z_n(\theta)-\frac{2n}{\theta}s_0}} - e^{-n^3}\geq \frac{1}{2}\mathbb{E} e^{-e^{\log Z_n(\theta)-\frac{2n}{\theta}s_0}},\label{eq:probineq1}
\end{align}
where we used \eqref{eq:LDPQntilde} in the last step to conclude that the first term dominates the second for large $n$.

\medskip

On the other hand, Chernoff's bound (or Markov's inequality) yields
\begin{align*}
	\mathbb{P} \left[\log Z_n(\theta)\le  \frac{2n}{\theta}s\right] 
	&=  \mathbb{P} \left[e^{-e^{\log Z_n(\theta)-\frac{2n}{\theta}s}}\ge e^{-  1}\right] \\ &\le e \mathbb{E}\left[e^{-e^{\log Z_n(\theta)-\frac{2n}{\theta}s}} \right]
\end{align*}
implying
\begin{align}\label{eq:probineq2}
	\log \mathbb{P} \left[\log Z_n(\theta)\le \frac{2n}{\theta}s\right] \le \log \mathbb{E}\left[e^{-e^{\log Z_n(\theta)-\frac{2n}{\theta}s}} \right]+1.
\end{align}
Hence, by combining \eqref{eq:probineq1} and \eqref{eq:probineq2} with \eqref{eq:LDPQntilde}, we obtain
\[\lim_{n\to\infty}\frac{1}{n^2}\log \mathbb{P} \left[\log Z_n(\theta)\le \frac{2n}{\theta}s\right]=-F(s,\theta),\]
which implies Conjecture \ref{conjecture:main}.

\section{Evidence for Ansatz \ref{ansatz1}}\label{section:ansatz}

In order to understand why Ansatz \ref{ansatz1} is plausible, we need to take a closer look at the signed biorthogonal measures $d\mu_n^\theta(x_1,\ldots, x_n)$ defined in \eqref{def:BiOM}--\eqref{def:BiOM2}.
As already mentioned in the introduction, it is convenient to consider the re-scaled biorthogonal measure
\be
d\widehat \mu_n^\theta(y_1,\ldots, y_n):= d\mu_n\left(\frac{2n}{\theta}y_1,\ldots, \frac{2n}{\theta}y_n\right)=\frac{1}{n!}\det\left(\widehat L_n^\theta(y_i,y_j)\right)_{i,j=1}^ndy_1\cdots dy_n,\ee
with $\widehat L_n^\theta$ given by \eqref{def:Lrescales}.

\paragraph{Macroscopic limit of the one-point function.}
For $\theta=0$, $d\widehat\mu_n^0$ is the eigenvalue distribution of the Wishart-Laguerre random matrix ensemble, and it is well-known that $\frac{1}{n}\widehat L_n^0(y,y)$ converges as $n\to\infty$ to the Marchenko-Pastur law, recall \eqref{eq:MP}. In other words, the particles $y_1,\ldots, y_n$ admit a limiting one-point function. There is no indication of a phase transition taking place at $\theta=0$, so we expect that the limit \[\lim_{n\to\infty}\frac{1}{n}\widehat L_n^\theta(y,y)=:h^\theta(y)\] exists also for $\theta>0$ sufficiently small, at least when $y>0$.
This can presumably be proved using a classical saddle point analysis of the double integral expression \eqref{def:Lrescales} for the kernel $\widehat L_n^\theta$. 
The Tracy-Widom limit \eqref{eq:TWcvgc} suggests in addition that the right-most endpoint of the support of $h^\theta$ is given by $-\theta\psi(\theta)$, and that $h^\theta(y)$ behaves like a square root as $y\to -\theta\psi(\theta)$. 
We know that
\[Q_n^\theta(s)=\int_{\mathbb R^n}\prod_{j=1}^n \left(1-1_{(s,+\infty)}(y_j)\right) d\widehat\mu_n^\theta(y_1,\ldots, y_n)=\int_{(-\infty,s)^n} d\widehat\mu_n^\theta(y_1,\ldots, y_n),\]
while
\[\widetilde Q_n^\theta(s)=\int_{\mathbb R^n}\prod_{j=1}^n \left(1-\sigma_{s,n,\theta}(y_j)\right) d\widehat\mu_n^\theta(y_1,\ldots, y_n),\qquad \sigma_{s,n,\theta}(y)=\frac{1}{1+e^{-\frac{2n}{\theta}(y-s)}}.\]
As $n\to\infty$, the function $1-\sigma_{s,n,\theta}(y)$ converges point-wise to $1-1_{(s,+\infty)}(y)$.
From a broad perspective, one may already expect at this point that $Q_n^\theta(s)$ and $\widetilde Q_n^\theta(s)$ have the same rate of decay, because the large deviation rate function is a global quantity associated to a determinantal point process that typically does not depend on microscopic deformations of the involved test function, such that one expects \[\lim_{n\to\infty}\frac{1}{n^2}\log \widetilde Q_n^\theta(s)=\lim_{n\to\infty}\frac{1}{n^2}\log Q_n^\theta(s),\]
which implies Ansatz \ref{ansatz1}.

\paragraph{Jacobi's identity.}
We will now develop some more concrete analytical estimates in support of Ansatz \ref{ansatz1}.
We define a deformation of $Q_n^\theta$ as follows: we define
\be
Q_n^\theta(s;t):=\det\left(1-\sigma_{s,n,\theta}^t\widehat L_n^\theta\right)_{L^2(\mathbb R)},\quad \sigma_{s,n,\theta}^t(y)=\frac{1}{1+e^{-\frac{2n}{\theta t}(y-s)}},
\ee
such that $Q_n^\theta(s,0)=Q_n^\theta(s)$ and $Q_n^\theta(s,1)=\widetilde Q_n^\theta(s)$.
By Jacobi's variational identity, we have
\begin{align*}
\partial_t\log Q_n(s;t)&={\rm Tr}\left(-(\partial_t\sigma_{s,n,\theta}^t)\widehat L_n^\theta\left(1-\sigma_{s,n,\theta}^t\widehat L_n^\theta\right)^{-1}\right).
\end{align*}
Abbreviating 
\be K_{s,n,\theta}^t:=\widehat L_n^\theta\left(1-\sigma_{s,n,\theta}^t\widehat L_n^\theta\right)^{-1},\ee
and writing $K_{s,n,\theta}^t(y,y')$ for a kernel of this operator, we have
\begin{align*}
\partial_t\log Q_n(s;t)&=\frac{2n}{\theta t^2}\int_{\mathbb R}\frac{e^{-\frac{2n}{\theta t}(y-s)}}{\left(1+e^{-\frac{2n}{\theta t}(y-s)}\right)^2} (y-s) K_{s,n,\theta}^t(y,y)dy.
\end{align*}

We can bound the integral on the right as follows:
\be\left|\partial_t\log Q_n(s;t)\right|\leq \frac{\theta}{2n}\sup_{y\in\mathbb R}\left|K_{s,n,\theta}^t(y,y)\right|\int_{\mathbb R}\frac{e^{-u}}{\left(1+e^{-u}\right)^2} |u| du.
\ee
We believe that the following holds, for reasons that we will explain below.
\begin{ansatz}\label{ansatz2}
We have the estimate
\be\sup_{y\in\mathbb R}\left|K_{s,n,\theta}^t(y,y)\right|=O(n^2),\qquad n\to\infty,\ee
uniformly for $s>\epsilon$, $0\leq t\leq 1$, and $0<\theta<\theta_0$.
\end{ansatz}

This would imply 
\be\left|\partial_t\log Q_n(s;t)\right|\leq \frac{C\theta n}{2}\int_{\mathbb R}\frac{e^{-u}}{\left(1+e^{-u}\right)^2} |u| du=O(\theta n),\qquad n\to\infty,
\ee
and hence by integrating between $t=0$ and $t=1$,
\be \log Q_n(s,1)-\log Q_n(s,0)=\log \widetilde Q_n(s)-\log Q_n(s)=O\left(\theta n\right) ,\qquad n\to\infty.
\ee
Combining this with Theorem \ref{theorem:main}, we obtain Ansatz \ref{ansatz1} and thus Conjecture \ref{conjecture:main}, assuming Ansatz \ref{ansatz2}.

\paragraph{Evidence for Ansatz \ref{ansatz2}.}
The kernel $K_{s,n,\theta}^t(y,y')$ appearing in Ansatz \ref{ansatz2} has itself an interpretation in terms of a biorthogonal measure. Whereas $\widehat L_n^\theta(y,y')$ is the kernel of the signed biorthogonal measure $d\widehat\mu_n^\theta$, $K_{s,n,\theta}^t$ is the kernel of a deformation of this biorthogonal measure. Consider the signed biorthogonal measure
\be
d\widehat \nu_{n,s,\theta}^t(y_1,\ldots, y_n):=\frac{1}{Z_n}\prod_{j=1}^n(1-\sigma_{n,s,\theta}^t(y_j))d\widehat \mu_n^\theta(y_1,\ldots, y_n),
\ee
where
\[Z_n=\int_{\mathbb R^n}\prod_{j=1}^n(1-\sigma_{n,s,\theta}^t(y_j))d\widehat \mu_n^\theta(y_1,\ldots, y_n).\]
It was proved in \cite[Section 4]{CC24} (as an application of a more general result from \cite{ClaeysGlesner}) that $K_{s,n,\theta}^t$ is the kernel of this determinantal signed measure with respect to the measure $(1-\sigma_{n,s,\theta}^t(y))dy$. In other words,
\be
d\widehat \nu_{n,s,\theta}^t(y_1,\ldots, y_n)=\frac{1}{n!}\det\left(K_{s,n,\theta}^t(y_j,y_k)\right)_{j,k=1}^n\prod_{j=1}^n(1-\sigma_{n,s,\theta}^t(y_j))d y_j,
\ee
and $K_{s,n,\theta}^t$ has the reproducing property
\[\int_{\mathbb R}K_{s,n,\theta}^t(y,u)K_{s,n,\theta}^t(u,y') (1-\sigma_{n,s,\theta}^t(u))du=K_{s,n,\theta}^t(y,y').\]

\medskip

Let us try to get some more insight in these deformed biorthogonal measures.
For $\theta=0$, we have
\be
d\widehat \nu_{n,s,\theta}^t(y_1,\ldots, y_n):=\frac{1}{Z_n}\prod_{1\leq j<i\leq n}(y_i-y_j)^2\ \prod_{j=1}^n(1-\sigma_{n,s,\theta}(y_j)) 1_{(0,+\infty)}(y_j)e^{-4ny_j}d y_j.
\ee
This is a deformation of the eigenvalue distribution of a Laguerre-Wishart random matrix, and it is clearly a positive measure. For $\theta=0$, $t=0$, it is the eigenvalue distribution of a Laguerre-Wishart random matrix under conditioning on the event that all eigenvalues $y_1,\ldots, y_n$ are smaller than $s$.
For $\theta=0$, $t>0$, the deformed biorthogonal ensemble can be constructed via a procedure of marking and conditioning as done in \cite{ClaeysGlesner}. The deformed ensemble then assigns a smaller likelihood to eigenvalue configurations with eigenvalues bigger than $s$, and one can think of it as a {\em pushed} Coulomb gas, following the terminology of \cite{CGKLDT, KrajenbrinkLeDoussal}.
Both for $t=0$ and $t>0$, the deformed ensembles are special cases of orthogonal polynomial ensembles, whose large $n$ asymptotic behavior is well understood. See e.g.\ \cite{Kuijlaars} for an overview of such results. In particular, it is understood that for $t=0$, the bulk of the eigenvalues lies for large $n$ between $0$ and $s$, and that the ensemble has a hard edge at $0$ and another hard edge at $s$, if $s<1$. Scaling limits of the eigenvalue correlation kernel $K_{s,n,\theta}^t(y,y')$ lead to the sine kernel in the bulk, and to the Bessel kernel near the hard edges. In particular, we have
\be K_{s,n,0}^0(y,y)=O(n),\qquad n\to\infty,\label{eq:boundK1}\ee
for $\epsilon<y<s-\epsilon$, and 
\be\label{eq:boundK2}K_{s,n,0}^0(y,y)=O(n^2),\qquad n\to\infty,\ee
for $y\leq \epsilon$ and for $s-\epsilon\leq y$. This gives Ansatz \ref{ansatz2} for $\theta=0$, $t=0$. As $t$ increases, the hard edge at $s$ turns into a soft edge, such that one expects that the above bounds still hold, and even that the bound near $y=s$ becomes non-optimal for bigger $t$.

\medskip

Let us now consider the case $\theta>0$. Then the deformed biorthogonal measure $d\widehat\nu_{n,s,\theta}^t$ is not necessarily positive; however we still expect that the one-point function $K_{s,n,\theta}^t(y,y)$ is positive for $y>\epsilon$, $\theta$ sufficiently small, and $n$ sufficiently large. 
  For $t=0$, the deformed measure $d\widehat\nu_{n,s,\theta}^0$ is the restriction of $d\widehat\mu_n^\theta(y_1,\ldots, y_n)$ to the region $(-\infty, s)^n$, renormalized to have total mass $1$. 
For $t>0$, the deformed measure $d\widehat\nu_{n,s,\theta}^t$ gives less (positive or negative) weight to configurations $(y_1,\ldots, y_n)$ for which $y_j$s lie in $(s,+\infty)$, so we can still think of it as a {\em pushed} measure, where weight is pushed towards $(-\infty,s)^n$.
We do not expect a qualitative difference between the correlation kernel $K_{s,n,0}^t(y,y)$ and $K_{s,n,\theta}^t(y,y)$ for $\theta>0$ sufficiently small; in particular, we expect (but cannot prove) that \eqref{eq:boundK1}--\eqref{eq:boundK2} continue to hold for $\theta>0$ sufficiently small and for $0<t<1$, such that Ansatz \ref{ansatz2} holds. 

Unfortunately, we have no explicit formulas or tools to analyze $K_{s,n,\theta}^t(y,y)$ asymptotically as $n\to\infty$: there is no double contour integral representation like for $L_n^\theta$, and there is no convenient RH characterization for it. We do however have access to a special case of this quantity, corresponding to $t=0$ and $y=s$. We indeed have by Jacobi's identity that
\[\partial_s\log Q_n(s;t)=K_{s,n,\theta}^0(s,s).
\]
Differentiating the result from Theorem \ref{theorem:main} on the other hand, we obtain that
\[K_{s,n,\theta}^0(s,s)=\partial_s\log Q_n(s;t)=O(n^2),\qquad n\to\infty,
\]
which is consistent with and confirms the order predicted in Ansatz \ref{ansatz2}. 

\paragraph{Comparison of Fredholm determinants.}
Alternatively, one may attempt to prove Ansatz \ref{ansatz1} by using a comparison bound for Fredholm determinants. For trace-class operators $A,B$, we have \cite[Chapter 5]{Simon}
\be\label{eq:Fredholmdiff}|\det(1+A)-\det(1+B)|\leq \|A-B\|_1e^{\max\{\|A\|_1,\|B\|_1\}+1},\ee
where $\|.\|_1$ denotes the trace norm.
 Setting $B=\mathcal L_n^s$ with $\mathcal L_n^s$ given by
 \[\left(\mathcal L_n^s\right)[f](x)=1_{(s,+\infty)}(y)e^{-cy}\int_{\mathbb R}\widehat L_n^\theta(y',y)e^{cy'}f(y')dy',\]
 like in the proof of Lemma \ref{lemma: det L = det H}, and $A$ given by
 \[\left(A\right)[f](y)=\sigma_{n,s,\theta}(y)e^{-cy}\int_{\mathbb R}\widehat L_n^\theta(y',y)e^{cy'}f(y')dy',\]
 the left hand side of \eqref{eq:Fredholmdiff} becomes 
 $|\widetilde Q_n(s)-Q_n(s)|$.
It is not immediate to compute the order of the right hand side: since the operators $A$ and $B$ are not Hermitian, the trace norm is not necessarily equal to the trace and thus hard to compute. One may nevertheless expect that $\|A-B\|_1=O(1/n)$ as $n\to\infty$. Unfortunately, for $0<s<-\theta\psi(\theta)$, one also expects that $\|A\|_1,\|B\|_1=O(n)$ as $n\to\infty$, such that the right hand side of \eqref{eq:Fredholmdiff} is not small. While the inequality \eqref{eq:Fredholmdiff} is sharp for small operators $A,B$, its quality becomes poor when $A,B$ are close to being singular, which is our situation.
This approach consequently does not lead to a proof nor to further support for Ansatz \ref{ansatz1}.

\paragraph{Conclusion.}

Let us summarize the line of arguments in support of Conjecture \ref{theorem:main}.
Above, we explained why we strongly believe that the technical Ansatz \ref{ansatz2} holds: firstly because it holds for $\theta=0$ and there is no sign of a dramatic change of behavior for $\theta>0$, and secondly because of the interpretation of the kernel $K_{s,n,\theta}^t(y,y')$ in terms of a deformed signed biorthogonal measure along with its interpretation as a pushed measure.
Then, we proved that the technical Ansatz \ref{ansatz2} implies the less technical Ansatz \ref{ansatz1}. Furthermore, we proved Theorem \ref{theorem:main}, and we proved that it implies Conjecture \ref{conjecture:main} under Ansatz \ref{ansatz1}.

\subsection*{Acknowledgements} The authors are grateful to Mattia Cafasso, Reda Chhaibi, Ivan Corwin, Yuchen Liao, and Jiyuan Zhang for useful discussions. They acknowledge support by {\em FNRS Research
Project T.0028.23} and by the {\em Fonds Sp\'ecial de Recherche} of UCLouvain.

\end{document}